\documentclass[]{IEEEtran}
\usepackage{latexsym}
\usepackage{amsfonts}
\usepackage{amsbsy}
\usepackage{amsmath,amssymb}
\usepackage{times}
\usepackage{graphicx}
\usepackage{enumerate}
\usepackage[usenames]{color}
\usepackage[dvips]{pstcol}
\usepackage{epstopdf}
\usepackage{cite}
\usepackage{amsmath}
\usepackage{amssymb}
\usepackage{amsfonts}
\usepackage{graphicx}
\usepackage{epsfig}
\usepackage{psfrag}
\usepackage{xcolor}
\usepackage{amsfonts, bm}
\usepackage{epstopdf}
\usepackage{cite}
\usepackage{color}
\usepackage{xcolor}
\usepackage{subfig}
\usepackage{booktabs}
\newtheorem{lemma}{Lemma}

\newtheorem{theorem}{Theorem}

\input epsf

\usepackage{lineno}
\usepackage{algorithm} 
\usepackage{multirow} 
\usepackage{algpseudocode}
\usepackage{graphics}
\usepackage{epsfig}
\usepackage{subeqnarray}
\usepackage{cases}
\usepackage{setspace}
\usepackage{diagbox}
\usepackage{graphicx}
\usepackage{subfig}
\usepackage{subfloat}
\allowdisplaybreaks[4]

\begin{document}
	
	\vspace{-2ex}
\captionsetup[figure]{name={Fig.},labelsep=period}
\captionsetup[table]{
	labelsep=newline,
	justification=centering,
	singlelinecheck=false,
	font=sc}
	
\title{Design and Performance Analysis of Wireless Legitimate Surveillance Systems with Radar Function}
\author{Mianyi Zhang,~\IEEEmembership{Student Member,~IEEE,} Yinghui He,~\IEEEmembership{Student Member,~IEEE,} Yunlong Cai,~\IEEEmembership{Senior Member,~IEEE,} Guanding Yu,~\IEEEmembership{Senior Member,~IEEE,} and Naofal Al-Dhahir,~\IEEEmembership{Fellow,~IEEE}
\vspace{-2ex}
\thanks{
The work of Y. Cai was supported in part by the National Natural Science Foundation of China under Grants 61971376, U22A2004, and 61831004.
The work of G. Yu was supported in part by GDNRC[2021]32. The work of N. Al-Dhahir was supported by Erik Jonsson Distinguished Professorship at UT-Dallas.
	
M. Zhang, Y. He, Y. Cai, and G. Yu are with the College of Information Science and Electronic Engineering, Zhejiang University, Hangzhou 310007, China, and also with Zhejiang Provincial Key Laboratory of Information Processing, Communication and Networking (IPCAN), Hangzhou 310007, China (email: \{mianyi\_zhang, 2014hyh, ylcai, yuguanding\}@zju.edu.cn).
	
N. Al-Dhahir is with the Department of Electrical and Computer Engineering, The University of Texas at Dallas, Richardson, TX 75080 USA (e-mail: aldhahir@utdallas.edu).
}
}

\maketitle
\vspace{-4ex}
\begin{abstract}
	Integrated sensing and communication (ISAC) has recently been considered as a promising approach to save spectrum resources and reduce hardware cost. Meanwhile, as information security becomes increasingly more critical issue, government agencies urgently need to legitimately monitor suspicious communications via proactive eavesdropping. Thus, in this paper, we investigate a wireless legitimate surveillance system with radar function. We seek to jointly optimize the receive and transmit beamforming vectors to maximize the eavesdropping success probability which is transformed into the difference of signal-to-interference-plus-noise ratios (SINRs) subject to the performance requirements of radar and surveillance. The formulated problem is challenging to solve. By employing the Rayleigh quotient and fully exploiting the structure of the problem, we apply the divide-and-conquer principle to divide the formulated problem into two subproblems for two different cases. For the first case, we aim at minimizing the total transmit power, and for the second case we focus on maximizing the jamming power. For both subproblems, with the aid of orthogonal decomposition, we obtain the optimal solution of the receive and transmit beamforming vectors in closed-form. Performance analysis and discussion of some insightful results are also carried out. Finally, extensive simulation results demonstrate the effectiveness of our proposed algorithm in terms of eavesdropping success probability.
	

\end{abstract}\label{key}

\vspace{0ex}

\section{Introduction}
With the rapid growth of the wireless communication industry, the carrier frequency in wireless communications has been pushed toward the frequency band commonly assigned to radar sensing systems due to the ever-increasing demand for higher communication rates \cite{survey1}. Additionally, beyond 5G and 6G wireless systems are required to provide high-accuracy sensing services in various applications \cite{survey2}, such as autonomous vehicle driving, robot navigation, and indoor localization for virtual reality. Correspondingly, future communication signals are desired to have high-resolution in both the time and angular domains, enabled by the millimeter wave (mmWave) and massive multiple-input multiple-output (MIMO) technologies. Therefore, integrated sensing and communication (ISAC) is considered to be a key technology in the next generation of wireless communication systems, because it not only can save spectrum resources, but also meet the demand for high-precision sensing using communication signals. To realize ISAC, the dual-functional radar-communication (DFRC) system is a promising approach since it can eliminate mutual interference and reduce hardware cost \cite{DF1,DF2,DF3,DF4}. It is noteworthy that different types of DFRC systems have been investigated in recent years, e.g., the waveform design for DFRC systems \cite{DF1}, \cite{RFse1}, robust orthogonal frequency division multiplexing (OFDM) DFRC systems \cite{ofdm1,ofdm2,ofdm3,ofdm4}, low probability of intercept (LPI) in DFRC systems \cite{lpi1,lpi2,lpi3,lpi4} and some other specific applications \cite{5g,wifi,amrfc}.

On the other hand, due to the open and broadcast nature of wireless signal transmission, information security is a common and critical challenge in wireless communications. Various physical layer security techniques as complements to traditional encryption have been widely studied to ensure information security of communications \cite{an,ab,key,sobf}. However, these techniques also benefit suspicious users who abuse communication resources to engage in illegal activities. Thus, it is essential for the government agencies to legitimately monitor suspicious communications \cite{gov}. To this end, passive eavesdropping is a basic approach, where the legitimate monitor directly listens to the suspicious links without transmitting any jamming signals. When the legitimate monitor is close to the suspicious transmitter, passive eavesdropping is efficient as the transmitted information can be decoded successfully. However, the legitimate monitor is generally far away from the suspicious transmitter to protect itself from exposure, which leads to a poor performance of passive eavesdropping due to the inferiority of the eavesdropping channel. To overcome this limitation, proactive eavesdropping via jamming is developed as an effective and legal approach for government agencies to prevent crimes. Specifically, the legitimate monitor purposely controls the suspicious communication rate via jamming to improve the eavesdropping efficiency \cite{pe1,pe2,pe3,pe4,peextra}. There has been a growing body of literature studying proactive eavesdropping under various system setups, including relay systems \cite{relay}, multi-antenna systems \cite{ma1,ma2,ma3}, unmanned aerial vehicle (UAV) assisted systems \cite{uav1}, wireless powered communication systems \cite{wpc}, and intelligent reflecting surface (IRS) enhanced wireless systems \cite{irs}.

As mentioned above, there is a strong need for government agencies to conduct legitimate eavesdropping. Moreover, for future legitimate monitor units, the sensing function will play a significant role in various applications, e.g., UAV flight tracking, autonomous vehicle driving, micro-robot navigation, etc. Inspired by these applications, it is important to study the joint design of proactive eavesdropping and radar sensing to improve the spectrum efficiency and meet the requirements of next-generation wireless communication systems. Several related works have been proposed in recent years. The authors of \cite{uav} considered a scenario where a legitimate UAV tracks suspicious UAV flights for preventing intended crimes. To enhance tracking accuracy, the legitimate UAV employing proactive eavesdropping and a tracking algorithm rather than radar sensing has been developed by utilizing eavesdropped information and channel state information (CSI). Meanwhile, the authors of \cite{secureradar} investigated the physical layer security in a MIMO DFRC system. This work assumes that a base station (BS) communicates with a number of legitimate users, and treats the radar target as a potential eavesdropper who can overhear the information from the BS. To ensure the transmission security, artificial noise has been employed to minimize the eavesdropping rate at the radar target.

However, the aforementioned studies do not focus on the joint design of proactive eavesdropping and radar sensing. In fact, to the best of our knowledge, there are few ISAC related papers for proactive eavesdropping in the literature. Motivated by this research gap, in this paper, we propose a DFRC system which integrates proactive eavesdropping and radar sensing. Specifically, we consider a full-duplex (FD) mode hardware-efficient legitimate monitor which is equipped with a large-scale antenna array and performs analog-digital (AD) hybrid beamforming. The hybrid beamforming structure consists of a baseband digital beamformer and a codebook-based analog beamformer. We seek to maximize the eavesdropping success probability by jointly optimizing the receive and transmit beamforming vectors subject to the performance requirements of radar and surveillance. We first transform the eavesdropping success probability into a difference of signal-to-interference-plus-noise ratios (SINRs) and formulate the optimization problem. Using the Rayleigh quotient and a deep analysis of the problem structure, we divide it into two subproblems for two different cases, which aim at minimizing the total transmit power and maximizing the jamming power. For both subproblems, we obtain the optimal receive and transmit beamforming vectors in closed-form with the aid of orthogonal decomposition. Then, the optimal beamforming algorithm is developed by combining these two cases. Moreover, we discuss some insightful results.

The specific contributions of this paper are summarized as follows.
\begin{itemize}
	\item[$\bullet$] We design a novel DFRC system which has not been studied before. To be specific, we introduce a wireless surveillance system with radar function, where the legitimate surveillance and radar sensing share the same spectrum and hardware. The legitimate monitor eavesdrops the suspicious transmission via smart jamming. At the same time, the legitimate monitor also works as a phased-array radar and sends probe signals to detect targets. Note that the legitimate monitor works in FD mode and is equipped with a large-scale antenna array. Additionally, the legitimate monitor performs the AD hybrid beamforming.
	
	\item[$\bullet$] We seek to optimize the receive and transmit beamforming vectors to maximize the eavesdropping success probability subject to the performance requirements of radar and surveillance. To make the problem solvable, we transform the probability into a difference of SINRs and formulate the optimization problem. Using the Rayleigh quotient, the optimal receive beamforming vectors are obtained and the optimization problem is simplified. However, the simplified problem is still highly coupled. Thus, by fully exploiting the structure of the simplified problem, we apply the divide-and-conquer principle to divide it into two subproblems for two cases. The first subproblem aims at minimizing the total transmit power and the second subproblem of the second case focuses on maximizing the jamming power.

	\item[$\bullet$] We propose an orthogonal basis construction based method to simplify the optimization problem for each case. Then, the problems can be converted into convex ones and the optimal receive and transmit beamforming vectors for each problem can be obtained in closed-form by analyzing the Karush-Kuhn-Tucker (KKT) conditions. By combining these two cases, the optimal beamforming algorithm is developed. Moreover, some insightful results regarding the optimal transmit power allocation are also discussed.

	\item[$\bullet$] We analyze the eavesdropping success probability according to the derived closed-form expressions of the optimal solutions. We find that there is a trade-off between the proactive eavesdropping performance and the radar detection performance and we quantify it. Finally, we provide extensive simulation results to verify the effectiveness of the proposed algorithm and analytical results.
\end{itemize}

The rest of this paper is organized as follows. Section II introduces the system model and formulates the corresponding problem for joint surveillance and radar sensing. The optimal algorithms for solving the formulated problem are introduced in Section III. Section IV provides a performance analysis and Section V presents the simulation results to validate the effectiveness of the proposed algorithm. Finally, Section VI concludes the paper.

{\it Notations}: In this paper, we use lower case letters, bold lower case letters, and bold upper case letters to denote scalars, vectors, and matrices, respectively. ${\left\| {\cdot} \right\|}$, $({\cdot})^{T}$, and $({\cdot})^{H}$ denote the Euclidean norm, transpose operator, and complex conjugate transpose operator, respectively. ${\mathbb C^{x \times y}}$ denotes the ${x \times y}$ complex space and $\bm{I}$ stands for the identity matrix. $\mathcal{CN}(\bm{0},\bm{A})$ denotes the distribution of a circularly symmetric complex Gaussian random variable with zero mean and covariance matrix $\bm{A}$. $\mathbb{E}\{x\}$ is the expectation of the random variable $x$ and $\mathcal{P}(\cdot)$ denotes the probability. Besides, $<{\bm \alpha}_1,{\bm \alpha}_2>$ represents the vector inner product of ${\bm \alpha}_1$ and ${\bm \alpha}_2$.


\section{System Model and Problem Formulation}

In this section, we introduce the proposed wireless legitimate surveillance system with radar function and formulate the optimization problem investigated in this paper.

\subsection{System Model}

\begin{figure}[htbp]
	\centering
	\includegraphics[scale=0.24]{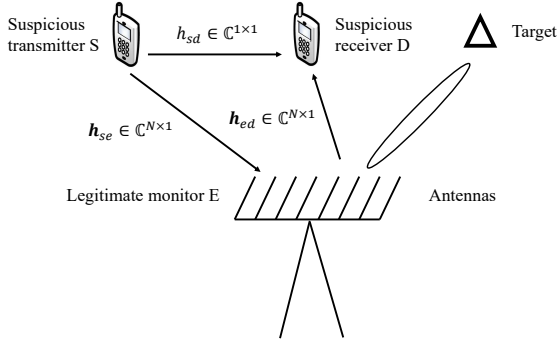}
		\vspace{0ex}
	\caption{A legitimate surveillance system with radar function.}
	\label{system}
	\vspace{-2ex}
\end{figure}

We consider a legitimate surveillance system including a pair of a suspicious transmitter and receiver and a legitimate monitor with radar function. Note that the identification of suspicious users can be implemented via big data analytics \cite{gov, identification1, identification2, identification3}, so it is regarded as prior information of our system. As shown in Fig. \ref{system}, the legitimate monitor E aims to proactively eavesdrop the suspicious transmission between transmitter S and receiver D. Meanwhile, since the radar function can contribute to the legitimate surveillance in several ways such as tracking the suspicious users, detecting potential obstacles, or assisting communications \cite{radar-assisted beamforming}, we assume that the legitimate monitor works as a phased-array radar and sends probe signals toward $N$ directions covering $180^{\circ}$ successively within a scan period. The suspicious transmitter and receiver are both equipped with a single antenna, and the legitimate monitor consisting of a large array of transmit and receive antennas. Besides, the legitimate monitor works in the FD mode to simultaneously eavesdrop and jam. Generally, we assume conventional far-field propagation with planar wavefront and quasi-static channel fading coefficients which remain unchanged during each transmission block. Let ${\bm h}_{se}\in {\mathbb C}^{N\times 1}, h_{sd}\in {\mathbb C}^{1\times 1}$, and ${\bm h}_{ed}\in {\mathbb C}^{N\times 1}$ denote the eavesdropping channel, suspicious channel, and jamming channel vectors, respectively. It should be noted that in practice, ${\bm h}_{se}$ and ${\bm h}_{ed}$ could be estimated by overhearing the pilot signals sent by S and D and $h_{s,d}$ could be obtained by overhearing the channel feedback sent from D to S \cite{csi_obtain}.

\vspace{-1ex}

\subsection{Transmit and Receive Process}

As shown in Fig. \ref{detailed}, both the transmit and receive modules of the legitimate monitor are equipped with $N$ antennas and $M$ $(M \ll N)$ radio frequency (RF) chains. We assume that the legitimate monitor detects $N$ directions successively within a scan period, i.e., it only probes one direction each time. Thus, $M-1$ RF chains are utilized for eavesdropping and only $1$ RF chain is employed for radar detection \cite{phased-array1}. For the transmit module, the transmit signal consists of the jamming signal and the probe signal for the $n$-th direction is given by
\begin{align}
{\bm x}_n(l)={\bm U}_n{\bm p}_{n,l} c_n(l),
\end{align}
where $c_n(l)$ denotes the transmit probe sequence with $\mathbb{E}\{|c_n(l)|^2\}=1$, $l=1,2,\ldots,L$ is the probe signal index and $L$ denotes the duration of each communication frame. Note that the jamming signal can share the same sequence with the probe signal. ${\bm U}_n=[{\bm u}_{n,1},{\bm u}_{n,2},\ldots,{\bm u}_{n,M}]\in {\mathbb C}^{N\times M}$ and ${\bm p}_{n,l}=[p_{n,l,1},p_{n,l,2},\ldots, p_{n,l,M}]^T\in {\mathbb C}^{M\times 1}$ denote the analog transmit beamforming matrix and the digital transmit beamforming vector when probing the $n$-th direction, respectively. Without loss of generality, we assume that the $M$-th RF chain is employed for radar detection, which means ${\bm u}_{n,1},\ldots,{\bm u}_{n,M-1}$ and $p_{n,l,1},\ldots,p_{n,l,M-1}$ are designed for proactive eavesdropping while ${\bm u}_{n,M}$ and $p_{n,l,M}$ are designed for radar detection. The analog beamforming matrix ${\bm U}_n$ is formulated based on a codebook. Similar to \cite{codebook1}, \cite{codebook2}, we adopt a discrete Fourier transformation (DFT) matrix as the codebook, denoted by ${\bm U} = [{{\bm \alpha}}({\theta}_1),\ldots, {{\bm \alpha}}({\theta}_N)]^T\in {\mathbb C}^{N\times N}$. ${\bm U}$ divides the space into $N$ beams where ${\bm \alpha}({\theta}_n) = [1,e^{j2\pi \frac{d}{\lambda}\text{sin}{\theta}_n},\ldots, $ $e^{j2\pi\frac{(N-1)d}{\lambda}\text{sin}{\theta}_n}]^T \in {\mathbb C}^{N\times 1}$ is the corresponding steering vector. Note that $d$ and $\lambda$ denote the antenna spacing and signal wavelength, respectively, and typically $d=\lambda /2$. Therefore, when probing ${\theta}_n$, ${\bm u}_{n,1},\ldots, {\bm u}_{n,M-1}$ are obtained by selecting the first $M-1$ steering vectors that maximize $|{{\bm \alpha}}^H({\theta}_n){\bm h}_{ed}|$ from the codebook, and we have ${\bm u}_{n,M}={{\bm \alpha}}^H({\theta}_n)$ to maximize the radar detection performance.

\begin{figure*}[htbp]
	\vspace{-10ex}
	\centering
	\includegraphics[scale=0.35]{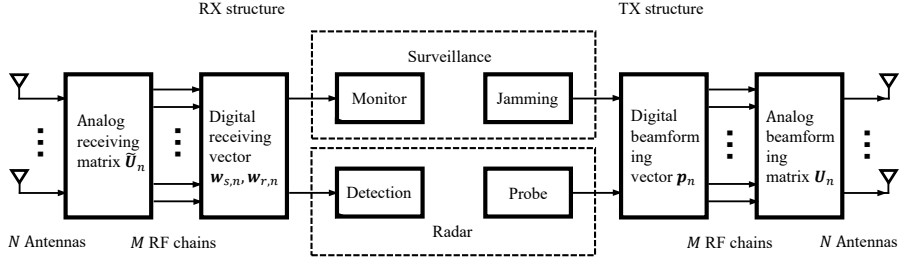}
	\vspace{-2ex}
	\caption{A detailed design of the legitimate monitor.}
	\label{detailed}
\end{figure*}

\begin{figure*}[htbp]
	\vspace{0ex}
	\centering
	\includegraphics[scale=0.35]{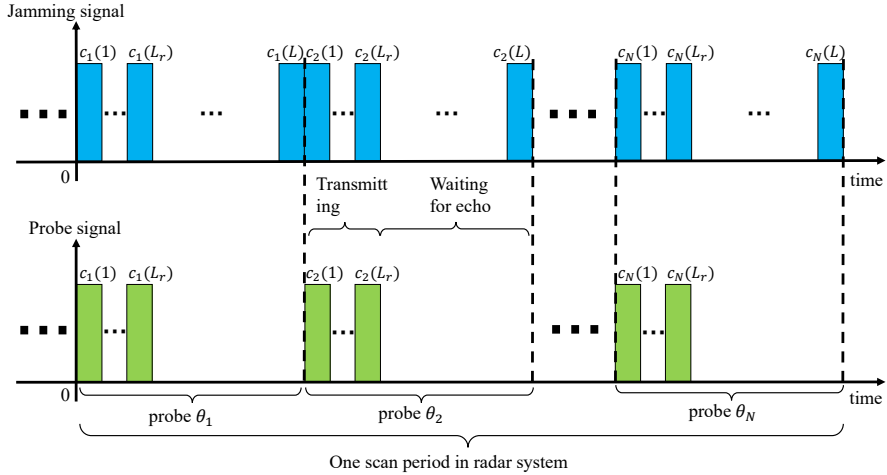}
		\vspace{0ex}
	\caption{Jamming signal and probe signal.}
	\label{time}
	\vspace{0ex}
\end{figure*}

As shown in Fig. \ref{time}, jamming lasts for a long time but one radar scan period is much shorter, while in radar detection direction ${\theta}_n$ changes rapidly. When probing ${\theta}_n$, the radar sends $L_r$ symbols for $l=1,\ldots, L_r$, leaving an interval to wait for the echo of the transmitted probe signal. Since ${\bm u}_{n,M}$ represents the beamforming vector for the radar detection, $p_{n,l,M}$ should be set as $0$ when $l=L_r+1,\ldots, L$ correspondingly, which means that there is no probe signal being transmitted. To distinguish ${\bm p}_{n,l}$ between these two different periods, we use ${\bm p}_{n,r}$ and ${\bm p}_{n,w}$ to denote ${\bm p}_{n,l}$ when $l \leq L_r$ and $L_r < l \leq L$, respectively, where ${\bm p}_{n,w}$ can be rewritten as ${\bm p}_{n,w}\triangleq {\bm \Sigma}\widetilde{\bm p}_{n,w}$ and ${\bm \Sigma} \in {\mathbb C}^{M\times M}$ denotes a diagonal matrix and its diagonal elements are all ones except the $M$-th element which is 0. Besides, the time ratios are defined as $\lambda_r\triangleq \frac{L_r}{L}$ and $\lambda_w\triangleq \frac{L-L_r}{L}$.

Then, the transmit signal of the legitimate monitor probing the $n$-th direction can be rewritten as
\begin{align}
{\bm x}_n(l)=
\begin{cases}
{\bm U}_n{\bm p}_{n,r} c_n(l) \quad& l \leq L_r,\\
{\bm U}_n{\bm \Sigma}\widetilde{\bm p}_{n,w} c_n(l) \quad& L_r < l \leq L.
\end{cases}
\end{align}

For the receive module, the received signals consisting of suspicious and radar echo signals are firstly processed by the analog receive beamforming matrix $\widetilde{{\bm U}}_n=[\widetilde{{\bm u}}_{n,1},\widetilde{{\bm u}}_{n,2},\ldots,\widetilde{{\bm u}}_{n,M}]\in {\mathbb C}^{N\times M}$. Similar to ${\bm U}_n$, when detecting ${\theta}_n$, the steering vectors corresponding to the first $M-1$ maximum elements of $|{\bm U}^H{\bm h}_{se}|$ are chosen to be $\widetilde{{\bm u}}_{n,1},\ldots,\widetilde{{\bm u}}_{n,M-1}$ and $\widetilde{{\bm u}}_{n,M}={{\bm \alpha}}({\theta}_n)$. Moreover, since the CSI of the self-interference link can be locally estimated at the legitimate monitor and the surveillance and radar share the same hardware and information, the self-interference can be significantly reduced to a low level based on certain interference cancellation techniques \cite{R1, R2, R3, R4, R5}. Hence, for simplicity, we assume that the self-interference at the legitimate monitor can be eliminated. Note that the same assumption has been widely adopted in other similar	works, such as \cite{R6}. Then, the received signal at the legitimate monitor E for the direction ${\theta}_n$ is given by
\begin{align}\label{y}
{\bm y}(l)=\underbrace{\widetilde{{\bm U}}_n^H{\bm h}_{se}\sqrt{p_s}s(l)}_{\text{suspicious signal}}+\underbrace{\widetilde{{\bm U}}_n^H{\bm A}_n{\bm U}_n{\bm p}_{n,l} c_n(l)}_{\text{radar echo}} \notag \\ +\underbrace{\widetilde{{\bm U}}_n^H\sum\limits_j{\bm A}_j{\bm U}_n{\bm p}_{n,l} c_n(l)}_{\text{radar clutter}}+{\bm n}(l),
\end{align}
where $s(l)$ denotes the $l$-th transmit symbol of the suspicious transmission with $\mathbb{E}\{|s(l)|^2\}=1$ and $p_s$ denotes the corresponding transmit power. Besides, ${\bm A}_n\triangleq {\beta}_n{{\bm \alpha}}({\theta}_n){{\bm \alpha}}^T({\theta}_n)\in \mathbb{C}^{N\times N}$ is the probing channel matrix for the potential target in direction $\theta_n$, where $\beta_n$ is the gain factor. Note that there are some non-target obstacles in all directions and let ${\bm A}_j$ denote the corresponding channel matrix of the $j$-th obstacle. ${\bm n}\sim \mathcal{CN}(\bm{0},\sigma^2 \bm{I})$ represents the additive white Gaussian noise (AWGN) with $\sigma^2$ denoting the noise power. Since high angular resolution of the radar detection can be achieved when $N$ is large, the clutter noise can be reduced to a relatively low level. Besides, due to the limited transmit power, the clutter noise cannot exceed its power upper bound. Thus, by employing the Central Limit Theorem, the sum of the clutters from all directions and the noise, i.e., $\widetilde{{\bm U}}_n^H\sum\limits_j{\bm A}_j{\bm U}_n{\bm p}_{n,l}+{\bm n}$, can be regarded as a Gaussian variable with the variance being $\widetilde{\sigma}^2$ \cite{clutter1, clutter2, clutter3}. Then, ${\bm y}$ is processed by the digital receive beamforming vectors ${\bm w}_{s,n}$ and ${\bm w}_{r,n}$ for the proactive eavesdropping and the radar detection, respectively.

Similarly, the received signal at the suspicious receiver D for direction ${\theta}_n$ can be written as
\begin{align}
 y_d(l)=\underbrace{h_{sd}\sqrt{p_s} s(l)}_{\text{suspicious signal}}+\underbrace{{\bm h}_{ed}^H{\bm U}_n{\bm p}_{n,l} c_n(l)}_{\text{jamming signal}}+n(l),
\end{align}
where $ n\sim \mathcal{CN}(0,\sigma^2)$ represents the AWGN with the noise variance being $\sigma^2$.


\subsection{SINR and Power Analysis}

To separate the suspicious and radar echo signals, ${\bm w}_{s,n}$ and ${\bm w}_{r,n}$ are designed in the null-spaces of $\widetilde{{\bm U}}_n^H{\bm A}_n{\bm U}_n$ and $\widetilde{{\bm U}}_n^H{\bm h}_{se}$, respectively. Therefore, we restructure ${\bm w}_{s,n}\triangleq {\bm Z}_{s,n}\widetilde{{\bm w}}_{s,n}$ and ${\bm w}_{r,n}\triangleq {\bm Z}_{r,n}\widetilde{{\bm w}}_{r,n}$, where we have $\widetilde{{\bm w}}_{s,n}$, $\widetilde{{\bm w}}_{r,n}\in {\mathbb C}^{(M-1)\times 1}$ denote the equivalent digital receive beamforming vectors, ${\bm Z}_{s,n}\triangleq \text{Null}\{\widetilde{{\bm U}}_n^H{\bm A}_n{\bm U}_n\}$, and ${\bm Z}_{r,n}\triangleq \text{Null}\{\widetilde{{\bm U}}_n^H{\bm h}_{se}\}\in {\mathbb C}^{M\times (M-1)}$. Note that ${\bm Z}_{s,n}$ and ${\bm Z}_{r,n}$ are normalized and always exist since the ranks of $\widetilde{{\bm U}}_n^H{\bm A}_n{\bm U}_n$ and $\widetilde{{\bm U}}_n^H{\bm h}_{se}$ are 1. With the null-space design, the principal part of the interference between the suspicious and radar echo signals can be eliminated.

Therefore, the SINRs of the legitimate monitor E, the suspicious receiver D, and the radar function are given as
\begin{align}
&\text{SINR}_\text{E} =\frac{\sum\limits^{N}_{n=1}|\widetilde{{\bm w}}_{s,n}^H{\bm Z}_{s,n}^H\widetilde{{\bm U}}_n^H{\bm h}_{se}|^2p_s}{\sum\limits^{N}_{n=1}\widetilde{\sigma}^2\widetilde{{\bm w}}_{s,n}^H\widetilde{{\bm w}}_{s,n}},\\
&\text{SINR}_\text{D} =
\!\frac{N|h_{sd}|^2p_s}{\sum\limits^{N}_{n=1}(\lambda_r|{\bm h}_{ed}^H{\bm U}_n{\bm p}_{n,r}|^2+\lambda_w|{\bm h}_{ed}^H{\bm U}_n{\bm \Sigma}\widetilde{\bm p}_{n,w}|^2)+N{\sigma}^2}\! ,\\
&\text{SINR}_\text{R} =\frac{|\widetilde{{\bm w}}_{r,n}^H{\bm Z}_{r,n}^H\widetilde{{\bm U}}_n^H{\bm A}_n{\bm U}_n{\bm p}_{n,r}|^2}{\widetilde{\sigma}^2\widetilde{{\bm w}}_{r,n}^H\widetilde{{\bm w}}_{r,n}}, \, \forall n=1,\ldots, N,
\end{align}
where $\text{SINR}_\text{E}$ and $\text{SINR}_\text{D}$ are computed over one scan period as shown in Fig. \ref{time}.

The total transmit power consumption can be written as
\begin{align}\label{total}
p_{\text{total}}=\sum\limits^{N}_{n=1}(\lambda_r\|{\bm p}_{n,r}\|^2+\lambda_w\|{\bm \Sigma}\widetilde{\bm p}_{n,w}\|^2).
\end{align}

\vspace{1ex}
\subsection{Problem Formulation}

In this paper, we aim at maximizing the probability of successful eavesdropping. However, if $\text{SINR}_\text{E}< \text{SINR}_\text{D}$, the legitimate monitor is not guaranteed to successfully eavesdrop the information. Therefore, we adopt the following indicator function to denote the event of successful eavesdropping as in \cite{pe4}
\begin{align}
Y=
\begin{cases}
	1,& \text{$\text{SINR}_\text{E}\ge \text{SINR}_\text{D}$},\\
	0,& \text{$\text{SINR}_\text{E}<\text{SINR}_\text{D}$}.
\end{cases}
\end{align}

Thus, we maximize $\mathbb{E}\{Y\}=\mathcal{P}(\text{SINR}_\text{E}\ge \text{SINR}_\text{D})$. To simplify the optimization objective, we can equivalently transform $\mathcal{P}(\text{SINR}_\text{E}\ge \text{SINR}_\text{D})$ into $\text{SINR}_\text{E} - \text{SINR}_\text{D}$ since they are positively correlated \cite{pe3, pe4}. Then, we can formulate the eavesdropping success probability maximization problem as
\begin{subequations} \label{p1}
	\begin{eqnarray}
		&\hspace{-2.5ex}\max\limits_{\left\{\substack{{\bm p}_{n,r}, \widetilde{\bm p}_{n,w},\\ \widetilde{{\bm w}}_{s,n},\widetilde{{\bm w}}_{r,n}}\right\}}&\hspace{-2ex} \frac{\sum\limits^{N}_{n=1}|\widetilde{{\bm w}}_{s,n}^H{\bm Z}_{s,n}^H\widetilde{{\bm U}}_n^H{\bm h}_{se}|^2p_s}{\sum\limits^{N}_{n=1}\widetilde{\sigma}^2\widetilde{{\bm w}}_{s,n}^H\widetilde{{\bm w}}_{s,n}} \notag \\
		&&\hspace{-4ex}-\frac{N|h_{sd}|^2p_s}{\sum\limits^{N}_{n=1}(\lambda_r|{\bm h}_{ed}^H{\bm U}_n{\bm p}_{n,r}|^2\hspace{-0.5ex}+\hspace{-0.5ex}\lambda_w|{\bm h}_{ed}^H{\bm U}_n{\bm \Sigma}\widetilde{\bm p}_{n,w}|^2)\hspace{-0.5ex}+\hspace{-0.5ex}N{\sigma}^2}\notag\\
		\\
		&\hspace{-2.5ex}\text{s.t.}&\hspace{-2ex}\sum\limits^{N}_{n=1}(\lambda_r\|{\bm p}_{n,r}\|^2+\lambda_w\|{\bm \Sigma}\widetilde{\bm p}_{n,w}\|^2)\leq p_{\text{max}},\\
		&&\hspace{-6ex}\frac{N|h_{sd}|^2p_s}{\sum\limits^{N}_{n=1}\hspace{-0.5ex}(\lambda_r|{\bm h}_{ed}^H{\bm U}_n{\bm p}_{n,r}|^2\hspace{-0.5ex}+\hspace{-0.5ex}\lambda_w|{\bm h}_{ed}^H{\bm U}_n{\bm \Sigma}\widetilde{\bm p}_{n,w}|^2)\hspace{-0.5ex}+\hspace{-0.5ex}N{\sigma}^2}\hspace{-0.5ex}\geq\hspace{-0.5ex} \gamma_s,\notag\\
		\\
		&&\hspace{-2ex}\frac{|\widetilde{{\bm w}}_{r,n}^H{\bm Z}_{r,n}^H\widetilde{{\bm U}}_n^H{\bm A}_n{\bm U}_n{\bm p}_{n,r}|^2}{\widetilde{\sigma}^2\widetilde{{\bm w}}_{r,n}^H\widetilde{{\bm w}}_{r,n}}\geq \gamma_r, \forall n=1,...,N,\notag \\
	\end{eqnarray}
\end{subequations}
where (\ref{p1}b) denotes the power constraint, (\ref{p1}c) guarantees the minimum essential monitoring rate and (\ref{p1}d) ensures the performance of radar detection. Specifically, the maximization of $\text{SINR}_\text{E}- \text{SINR}_\text{D}$ means that in the proactive eavesdropping system, the legitimate monitor allocates as much power as possible to the jamming signal towards the suspicious receiver. However, when the legitimate monitor successfully eavesdrops the information, the monitoring rate is equal to the transmission rate of the suspicious link \cite{csi_obtain}, which means that with the decrease of $\text{SINR}_\text{D}$, the transmission rate of the suspicious link decreases, i.e., the monitoring rate decreases. Thus, to maintain an essential monitoring rate, constraint (\ref{p1}c) guarantees the minimum of $\text{SINR}_\text{D}$.

Problem (\ref{p1}) is difficult to solve since (\ref{p1}b), (\ref{p1}c), and (\ref{p1}d) are highly coupled for ${\bm p}_{n,r}$ and $ \widetilde{\bm p}_{n,w}$. In the following section, we first obtain closed-form solutions for $\widetilde{{\bm w}}_{s,n}$ and $\widetilde{{\bm w}}_{r,n}$. To proceed, we fully exploit the structure of the problem and equivalently divide it into two subproblems based on two cases.


\section{Proposed Algorithm}

In this section, we first employ the Rayleigh quotient to obtain the optimal digital receive beamforming vectors and simplify the optimization problem. Then, by fully exploiting the structure of the simplified problem, we apply the divide-and-conquer principle to divide it into two subproblems. The first subproblem aims at minimizing the total transmit power and the second subproblem focuses on maximizing the jamming power. We propose an orthogonal basis construction based method to further simplify each subproblem. By checking the KKT conditions, the optimal digital transmit beamforming vectors for each subproblem can be obtained in closed-form. Then, we obtain the optimal beamforming algorithm by combining the solutions of these two subproblems.

	\vspace{-1.5ex}

\subsection{Reformulation of Problem (\ref{p1})}
By observing problem (\ref{p1}), $\widetilde{{\bm w}}_{s,n}$ only appears in the objective function. Thus, we decompose problem (\ref{p1}) and obtain the following problem for maximizing $\text{SINR}_\text{E}$ :
\begin{align}\label{p2}
	\max\limits_{\widetilde{{\bm w}}_{s,n}}\frac{\sum\limits^{N}_{n=1}|\widetilde{{\bm w}}_{s,n}^H{\bm Z}_{s,n}^H\widetilde{{\bm U}}_n^H{\bm h}_{se}|^2p_s}{\sum\limits^{N}_{n=1}\widetilde{\sigma}^2\widetilde{{\bm w}}_{s,n}^H\widetilde{{\bm w}}_{s,n}}.
\end{align}

Defining $\widetilde{\bm w}_{s}\triangleq [\widetilde{{\bm w}}_{s,1}^T,\ldots,\widetilde{{\bm w}}_{s,N}^T]^T \in {\mathbb C}^{(M-1)N\times 1}$ and $\widetilde{\bm h}_{s}\triangleq [({\bm Z}_{s,1}^H\widetilde{{\bm U}}_1^H{\bm h}_{se})^T,\ldots,({\bm Z}_{s,N}^H\widetilde{{\bm U}}_N^H{\bm h}_{se})^T]^T \in {\mathbb C}^{(M-1)N\times 1}$, problem (\ref{p2}) can be transformed to
\begin{align}\label{pp12}
	\max\limits_{\widetilde{\bm w}_{s}}\frac{|\widetilde{\bm w}_{s}^H\widetilde{\bm h}_{s}|^2}{\widetilde{\sigma}^2\widetilde{\bm w}_{s}^H\widetilde{\bm w}_{s}}.
\end{align}

Problem (\ref{pp12}) is a standard Rayleigh quotient maximization problem. Thus, we have following key result:
\begin{lemma} \label{lemma1}
	According to the Rayleigh quotient \cite[pp. 301--304]{rq}, a closed-form expression for the optimal $\widetilde{\bm w}_{s}$ is given by
	\begin{align}\label{ws}
	\widetilde{\bm w}_{s} = \frac{\widetilde{\bm h}_{s}}{\widetilde{\sigma}^2}.
	\end{align}
\end{lemma}

Back to problem (\ref{p1}), apparently, $\widetilde{{\bm w}}_{r,n}$ only exists in constraint (\ref{p1}d) corresponding to the radar performance. Similar to $\widetilde{\bm w}_{s}$, we have following result:
\begin{lemma} \label{lemma2}
	The optimal $\widetilde{{\bm w}}_{r,n}$ in terms of ${\bm p}_{n,r}$ is given by
	\begin{align}\label{wr}
	\widetilde{\bm w}_{r,n} = \frac{{\bm Z}_{r,n}^H\widetilde{{\bm U}}_n^H{\bm A}_n{\bm U}_n{\bm p}_{n,r}}{\widetilde{\sigma}^2}.
	\end{align}
\end{lemma}

With (\ref{ws}) and (\ref{wr}), problem (\ref{p1}) can be reformulated to the following problem:
\begin{subequations}\label{main}
	\begin{eqnarray}
		&\hspace{-3ex}\min\limits_{\left\{\substack{{\bm p}_{n,r}, \\ \widetilde{\bm p}_{n,w}}\right\}}&\hspace{-2.5ex}\!\frac{N|h_{sd}|^2p_s}{\sum\limits^{N}_{n=1}(\lambda_r|{\bm h}_{ed}^H{\bm U}_n{\bm p}_{n,r}|^2\hspace{-0.5ex}+\hspace{-0.5ex}\lambda_w|{\bm h}_{ed}^H{\bm U}_n{\bm \Sigma}\widetilde{\bm p}_{n,w}|^2)\hspace{-0.5ex}+\hspace{-0.5ex}N{\sigma}^2}\! \notag \\
		\\
		&\hspace{-3ex}\text{s.t.}&\hspace{-2.5ex}\sum\limits^{N}_{n=1}(\lambda_r\|{\bm p}_{n,r}\|^2+\lambda_w\|{\bm \Sigma}\widetilde{\bm p}_{n,w}\|^2)\leq p_{\text{max}},\\
		&&\hspace{-2.5ex}\!\frac{N|h_{sd}|^2p_s}{\sum\limits^{N}_{n=1}(\lambda_r|{\bm h}_{ed}^H{\bm U}_n{\bm p}_{n,r}|^2\hspace{-0.5ex}+\hspace{-0.5ex}\lambda_w|{\bm h}_{ed}^H{\bm U}_n{\bm \Sigma}\widetilde{\bm p}_{n,w}|^2)\hspace{-0.5ex}+\hspace{-0.5ex}N{\sigma}^2}\hspace{-0.5ex}\geq \hspace{-0.5ex} \gamma_s,\!\notag \\
		\\
		&&\hspace{-2.5ex}\|{\bm Z}_{r,n}^H\widetilde{{\bm U}}_n^H{\bm A}_n{\bm U}_n{\bm p}_{n,r}\|^2\geq \gamma_r{\widetilde{\sigma}^2}, \forall n=1,...,N,\notag \\
	\end{eqnarray}
\end{subequations}
where only $\text{SINR}_\text{D}$ is left in (\ref{main}a) since $\text{SINR}_\text{E}$ is maximized as a constant by applying (\ref{ws}).

Although only ${\bm p}_{n,r}$ and $\widetilde{\bm p}_{n,w}$ remain to be optimized, the problem is still challenging to be solved. Therefore, in the next subsection, we will discuss the classification of problem (\ref{main}).

	
\subsection{Classification Discussion}
Based on problem (\ref{main}), a significant contradiction exists between (\ref{main}a) and (\ref{main}c). Note that (\ref{main}a) requires the minimization of $\text{SINR}_\text{D}$, while (\ref{main}c) provides a lower bound constraint for $\text{SINR}_\text{D}$. Defining $p_{\text{th}}$ as the minimum total transmit power consumption, when $\text{SINR}_\text{D}$ reaches its lower bound $\gamma_s$ and the performance of radar detection is maintained. If $p_{\text{max}}\geq p_{\text{th}}$, the jamming signal should be allocated as much power as possible until $\text{SINR}_\text{D}$ reaches its lower bound $\gamma_s$. However, if $p_{\text{max}}< p_{\text{th}}$, the legitimate monitor will utilize all the power to jam the suspicious receiver and maintain the performance of radar detection. Therefore, this suggests that problem (\ref{main}) can be classified into two cases according to the relationship between $p_{\text{max}}$ and $p_{\text{th}}$. Without loss of generality, we first assume that $p_{\text{max}}\geq p_{\text{th}}$, then the monitoring rate constraint (\ref{main}c) becomes an equality constraint and the optimization objective in (\ref{main}a) turns into a constant equal to $\gamma_s$. Hence, we change the optimization objective into the total transmit power consumption in (\ref{main}b), and problem (\ref{main}) can be simplified into a power minimization problem:
\begin{subequations}\label{power}
	\begin{eqnarray}
		&\hspace{-2ex}\min\limits_{\{{\bm p}_{n,r}, \widetilde{\bm p}_{n,w}\}}&\sum\limits^{N}_{n=1}(\lambda_r\|{\bm p}_{n,r}\|^2+\lambda_w\|{\bm \Sigma}\widetilde{\bm p}_{n,w}\|^2)\\
		&\hspace{-2ex}\text{s.t.}&\sum\limits^{N}_{n=1}(\lambda_r|{\bm h}_{ed}^H{\bm U}_n{\bm p}_{n,r}|^2+\lambda_w|{\bm h}_{ed}^H{\bm U}_n{\bm \Sigma}\widetilde{\bm p}_{n,w}|^2)\notag \\
		&& =\frac{N|h_{sd}|^2p_s}{\gamma_s}-N{\sigma}^2,\\
		&&\|{\bm Z}_{r,n}^H\widetilde{{\bm U}}_n^H{\bm A}_n{\bm U}_n{\bm p}_{n,r}\|^2\geq \gamma_r{\widetilde{\sigma}^2}, \forall n=1,...,N.\notag \\
	\end{eqnarray}
\end{subequations}

After solving the power minimization problem (\ref{power}), $p_{\text{th}}$ can be obtained. Comparing $p_{\text{th}}$ with $p_{\text{max}}$, if $p_{\text{max}}\geq p_{\text{th}}$ is verified, then the optimal solution to the reformulated problem (\ref{main}) is obtained. However, if, on the contrary, $p_{\text{max}}< p_{\text{th}}$, it indicates that $\text{SINR}_\text{D}$ cannot be reduced to its lower bound $\gamma_s$. In this case, to fully utilize the finite power, the legitimate monitor allocates power to jamming signals as much as possible. Thus, the monitoring rate constraint (\ref{main}c) is always satisfied and the constraint (\ref{main}b) regarding the total transmit power consumption becomes an equality constraint. Then, we can transform problem (\ref{main}) into a jamming power maximization problem:
\begin{subequations}\label{jam}
	\begin{align}
	\max\limits_{\{{\bm p}_{n,r}, \widetilde{\bm p}_{n,w}\}}& \sum\limits^{N}_{n=1}(\lambda_r|{\bm h}_{ed}^H{\bm U}_n{\bm p}_{n,r}|^2+\lambda_w|{\bm h}_{ed}^H{\bm U}_n{\bm \Sigma}\widetilde{\bm p}_{n,w}|^2)\\
	\text{s.t.}\,\,&\sum\limits^{N}_{n=1}(\lambda_r\|{\bm p}_{n,r}\|^2+\lambda_w\|{\bm \Sigma}\widetilde{\bm p}_{n,w}\|^2)= p_{\text{max}},\\
	&\|{\bm Z}_{r,n}^H\widetilde{{\bm U}}_n^H{\bm A}_n{\bm U}_n{\bm p}_{n,r}\|^2\geq \gamma_r{\widetilde{\sigma}^2}, \, \forall n=1,\ldots, N.
	\end{align}
\end{subequations}

In summary, we classify problem (\ref{main}) into a power minimization problem (\ref{power}) and a jamming power maximization problem (\ref{jam}) corresponding to two different cases based on the relationship between $p_{\text{max}}$ and $p_{\text{th}}$. To solve the original problem (\ref{main}), we first obtain the optimal solution of problem (\ref{power}). Then, the minimum of the total transmit power consumption $p_{\text{th}}$ is compared with $p_{\text{max}}$. If $p_{\text{max}}\geq p_{\text{th}}$, the optimal solution to problem (\ref{main}) is obtained. Otherwise, this indicates that the original problem (\ref{main}) should be classified into the jamming power maximization problem (\ref{jam}). Thus, the optimal solution to problem (\ref{main}) is obtained by solving the jamming power maximization problem (\ref{jam}).

In the following subsection, we propose an orthogonal basis construction based algorithm to obtain the optimal solutions for the power minimization and jamming power maximization problems, respectively.

	\vspace{-3ex}
	
\subsection{Solution to the Power Minimization Problem}

Regarding (\ref{power}b), the left hand side (LHS) can be expressed via vector inner product as
\begin{subequations}\label{ip1}
	\begin{align}
	|{\bm h}_{ed}^H{\bm U}_n{\bm p}_{n,r}|=|<{\bm p}_{n,r},{\bm U}_n^H{\bm h}_{ed}>|,\\
	|{\bm h}_{ed}^H{\bm U}_n{\bm \Sigma}\widetilde{\bm p}_{n,w}|=|<\widetilde{\bm p}_{n,w},{\bm \Sigma}{\bm U}_n^H{\bm h}_{ed}>|.
	\end{align}
\end{subequations}

Similarly, with ${\bm A}_n={\beta}_n{{\bm \alpha}}({\theta}_n){{\bm \alpha}}^T({\theta}_n)$, the LHS of (\ref{power}c) can be rewritten as
\begin{align}\label{ip2}
\|{\bm Z}_{r,n}^H\widetilde{{\bm U}}_n^H{\bm A}_n{\bm U}_n{\bm p}_{n,r}\|^2\hspace{-0.5ex}=\hspace{-0.5ex}{\beta}_n^2g_n|\hspace{-0.5ex}<{\bm p}_{n,r},{\bm U}_n^H{{\bm \alpha}}^*({\theta}_n)>\hspace{-0.5ex}|^2,
\end{align}
where $g_n\triangleq {{\bm \alpha}}^H({\theta}_n)\widetilde{{\bm U}}_n{\bm Z}_{r,n}{\bm Z}_{r,n}^H\widetilde{{\bm U}}_n^H{{\bm \alpha}}({\theta}_n)$.

By observing (\ref{ip1}) and (\ref{ip2}), the vectors are of high significance as they represent the energy transmission directions. Hence, we define the following normalized direction bases:
\begin{subequations}\label{bases}
\begin{align}
	&{\bm v}_{\text{sum},n} = \frac{{\bm U}_n^H{\bm h}_{ed}}{\|{\bm U}_n^H{\bm h}_{ed}\|}, \\
	&{\bm v}_{\text{radar}} = \frac{{\bm U}_n^H{{\bm \alpha}}^*({\theta}_n)}{\|{\bm U}_n^H{{\bm \alpha}}^*({\theta}_n)\|}, \\
	&{\bm v}_{\text{jam}} = \frac{{\bm \Sigma}{\bm U}_n^H{\bm h}_{ed}}{\|{\bm \Sigma}{\bm U}_n^H{\bm h}_{ed}\|}.
\end{align}
\end{subequations}
The inner products between these direction bases can be written as
\begin{subequations}\label{basesip}
\begin{align}
	&<{\bm v}_{\text{sum},n},{\bm v}_{\text{radar}}> = \sqrt{\frac{g_{\text{radar},n}}{g_{\text{sum},n}}},\\ &<{\bm v}_{\text{sum},n},{\bm v}_{\text{jam}}> = \sqrt{\frac{g_{\text{jam}}}{g_{\text{sum},n}}}, \\
	&<{\bm v}_{\text{radar}},{\bm v}_{\text{jam}}> = 0,
\end{align}
\end{subequations}
where $	g_{\text{sum},n} = {\|{\bm U}_n^H{\bm h}_{ed}\|}^2$, $g_{\text{radar},n} = {|{{\bm \alpha}}^T({\theta}_n){\bm h}_{ed}|}^2$, $g_{\text{jam}} = {\|{\bm \Sigma}{\bm U}_n^H{\bm h}_{ed}\|}^2$, and we should note that $g_{\text{sum},n} = g_{\text{radar},n} + g_{\text{jam}}$.

The result of those inner products indicates that ${\bm v}_{\text{radar}}$ and ${\bm v}_{\text{jam}}$ are two orthogonal bases of the subspace determined by them. Besides, as $<{\bm v}_{\text{sum},n},{\bm v}_{\text{radar}}>^2 + <{\bm v}_{\text{sum},n},{\bm v}_{\text{jam}}>^2=1$, ${\bm v}_{\text{sum},n}$ is a linear combination of these two orthogonal bases which can be expressed as
\begin{align}\label{vsum}
	{\bm v}_{\text{sum},n} = \sqrt{\frac{g_{\text{jam}}}{g_{\text{sum},n}}}{\bm v}_{\text{jam}}+\sqrt{\frac{g_{\text{radar},n}}{g_{\text{sum},n}}}{\bm v}_{\text{radar}}.
\end{align}

	\begin{figure}[t]
		\vspace{-4ex}
		\centering
		\includegraphics[scale=0.25]{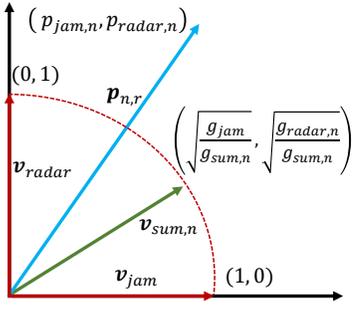}
		\vspace{-2ex}
		\caption{${\bm p}_{n,r}$ is a linear combination of ${\bm v}_{\text{jam}}$ and ${\bm v}_{\text{radar}}$.}
		\label{pfig}
		\vspace{-2ex}
	\end{figure}
	
According to (\ref{vsum}), an intuitive idea to simplify problem (\ref{power}) is decomposing ${\bm p}_{n,r}$ and $\widetilde{\bm p}_{n,w}$ into ${\bm v}_{\text{radar}}$, ${\bm v}_{\text{jam}}$, and other components which are orthogonal to the subspace determined by the two orthogonal bases. Thus, we have the following key result.
\begin{lemma} \label{lemma3}
	In problem (\ref{power}), to minimize the objective, ${\bm p}_{n,r}$ has no other components orthogonal to the subspace determined by ${\bm v}_{\text{radar}}$ and ${\bm v}_{\text{jam}}$. Meanwhile, it is always optimal to allocate no power to $\widetilde{\bm p}_{n,w}$, which means $\widetilde{\bm p}_{n,w}=\bm{0},\forall n=1,\ldots, N.$
\end{lemma}
\begin{proof}
	See Appendix A.
\end{proof}

Lemma 3 shows that in problem (\ref{power}), the analog transmit beamformer only concentrates energy on the jamming direction and the detection direction. Therefore, in this problem only ${\bm p}_{n,r}$ remains to be optimized, which can be decomposed as
\begin{align}\label{pdcmp2}
{\bm p}_{n,r}=p_{\text{jam},n}{\bm v}_{\text{jam}}+p_{\text{radar},n}{\bm v}_{\text{radar}},
\end{align}
where $p_{\text{jam},n}$ and $p_{\text{radar},n}$ are scalars denoting the modulus of each direction as shown in Fig. \ref{pfig}.

Based on (\ref{pdcmp2}), we can simplify problem (\ref{power}) as follows:
\begin{subequations}\label{power2}
	\begin{align}
	\min\limits_{\{p_{\text{jam},n},p_{\text{radar},n}\} }&\sum\limits^{N}_{n=1}\lambda_r(p_{\text{jam},n}^2+p_{\text{radar},n}^2)\\
	\text{s.t.}\,\,&\sum\limits^{N}_{n=1}\lambda_r(p_{\text{jam},n}^2g_{\text{jam}}+p_{\text{radar},n}^2g_{\text{radar},n})\notag \\
	&= \frac{N|h_{sd}|^2p_s}{\gamma_s}-N{\sigma}^2,\\
	&p_{\text{radar},n}^2\geq {\frac{\gamma_r{\widetilde{\sigma}^2}}{{\beta}_n^2g_n}}, \, \forall n=1,\ldots, N,\\
	&p_{\text{jam},n}\geq 0,p_{\text{radar},n}\geq 0, \, \forall n=1,\ldots, N.
	\end{align}
\end{subequations}

For problem (\ref{power2}), by transforming it into a convex problem and applying the KKT conditions, we have the following important result.
\begin{theorem} \label{theorem1}
	In the power minimization problem (\ref{power2}), the closed-form expressions of optimal $p_{\text{jam},n}$ and $p_{\text{radar},n}$ are given by
	\begin{subequations}\label{close1}
		\begin{eqnarray}
			&\hspace{-2ex}p_{\text{jam},n}=& \sqrt{\frac{|h_{sd}|^2p_s-{\sigma}^2\gamma_s}{\gamma_s\lambda_rg_{\text{jam}}}-\frac{1}{Ng_{\text{jam}}}\sum\limits^{N}_{n=1}\frac{\gamma_r{\widetilde{\sigma}^2}g_{\text{radar},n}}{{\beta}_n^2g_n}},\notag \\
			&&\forall n=1,...,N,\\
			&\hspace{-2ex}p_{\text{radar},n}=& \sqrt{ \frac{\gamma_r{\widetilde{\sigma}^2}}{{\beta}_n^2g_n}}, \forall n=1,...,N.
		\end{eqnarray}
	\end{subequations}
\end{theorem}
\begin{proof}
	See Appendix B.
\end{proof}

By observing (\ref{close1}), since $p_{\text{jam},n}$ and $p_{\text{radar},n}$ represent the power allocated for the jamming and probe signals, respectively, (\ref{close1}b) indicates that in the power minimization problem, the probe signal is allocated an amount of power that exactly satisfies the SINR constraint of the radar detection performance. As the constraint of radar SINR becomes tighter, the noise interference increases, and the radar channel gain decreases, the system tends to allocate more power to the probe signal but still guarantees the equivalence of the radar SINR requirement. Meanwhile, (\ref{close1}a) shows that the probe signal also play a role in jamming the suspicious receiver. The reason is that, generally, the jamming channel is not necessarily orthogonal to the radar channel, which leads to the power leakage from the probe signal to the jamming channel. Besides the power leakage, the legitimate monitor additionally allocates as much power as possible to minimize the SINR at the suspicious receiver until it reaches the lower bound. The closed-form expression of $p_{\text{jam},n}$ reveals that the power allocated for the jamming signal is not only dependent on the suspicious channel, the suspicious signal transmit power, the minimum SINR requirement at the suspicious receiver, and the jamming channel gain, but also affected by the radar SINR constraint, the gain of the power leakage from the probe signal, and the radar channel gain indirectly.

Although the power minimization problem (\ref{power}) is solved and its optimal solution is obtained, it is essential to verify if the original problem (\ref{main}) can be classified as the power minimization problem. Thus, we have the following theorem:
\begin{theorem} \label{theorem2}
	In the power minimization problem (\ref{power}), the minimum total transmit power consumption is given by
\begin{align} \label{ptotal}
p_{\text{th}} = &\frac{N|h_{sd}|^2p_s-N{\sigma}^2\gamma_s}{\gamma_sg_{\text{jam}}}-\frac{\lambda_r}{g_{\text{jam}}}\sum\limits^{N}_{n=1}\frac{\gamma_r{\widetilde{\sigma}^2}g_{\text{radar},n}}{{\beta}_n^2g_n}\notag \\
&+\lambda_r\sum\limits^{N}_{n=1}\frac{\gamma_r{\widetilde{\sigma}^2}}{{\beta}_n^2g_n}.
\end{align}
\end{theorem}
\begin{proof}
	By applying (\ref{total}), (\ref{pdcmp2}), and (\ref{close1}), the closed-form expression of $p_{\text{th}}$ is obtained.
\end{proof}

If $p_{\text{max}} > p_{\text{th}}$, then problem (\ref{main}) can be transformed into the power minimization problem (\ref{power}) and (\ref{close1}) is also the optimal solution to the original problem. Otherwise, this indicates that problem (\ref{main}) is equivalent to the jamming power maximization problem (\ref{jam}), in which the legitimate monitor makes full use of the finite power for jamming.

Regarding $p_{\text{th}}$, it is the sum of the radar transmission power and the jamming power. Since generally $g_{\text{jam}}$ is much greater than $g_{\text{radar},n}$, the power leakage $\dfrac{\lambda_r}{g_{\text{jam}}}\sum\limits^{N}_{n=1}\dfrac{\gamma_r{\widetilde{\sigma}^2}g_{\text{radar},n}}{{\beta}_n^2g_n}$ can be neglected. Thus, the threshold $p_{\text{th}}$ is affected by the lower bound of $\text{SINR}_\text{D}$ $\gamma_s$ and the lower bound of $\text{SINR}_\text{R}$ $\gamma_r$. With the decrease of $\gamma_s$ or the increase of $\gamma_r$, more power should be allocated for jamming or probing, which will lead to a high threshold. However, as mentioned above, only when the threshold $p_{\text{th}}$ is lower than $p_{\text{max}}$, can the original problem be classified into the power minimization problem (\ref{power}). Therefore, if $\gamma_s$ is too tight or $\gamma_r$ is too loose, the original problem (\ref{main}) is more likely to be classified into the jamming power maximization problem (\ref{jam}), which will be solved by the solution proposed in the following subsection.

	\vspace{-2.5ex}

\subsection{Solution to the Jamming Power Maximization Problem}

As the jamming power maximization problem (\ref{jam}) has a similar structure to the power minimization problem (\ref{power}), by constructing orthogonal bases based on (\ref{bases}) and (\ref{basesip}), we have the following theorem.
\begin{theorem} \label{theorem3}
	In the jamming power maximization problem (\ref{jam}), ${\bm p}_{n,r}$ has no other components orthogonal to the subspace determined by the set of direction bases. Meanwhile, it is always optimal to set $\widetilde{\bm p}_{n,w}=\bm{0},\forall n=1,\ldots, N$. Let us decompose ${\bm p}_{n,r}$ as follows:
	\begin{align}\label{pdcmp3}
	{\bm p}_{n,r}=\widetilde{p}_{\text{jam},n}{\bm v}_{\text{jam}}+\widetilde{p}_{\text{radar},n}{\bm v}_{\text{radar}}.
	\end{align}
	
	Then, the closed-form expressions of the optimal coefficients for ${\bm p}_{n,r}$ is given by
	\begin{subequations}\label{close2}
		\begin{align}
		&\widetilde{p}_{\text{jam},n}= \sqrt{\frac{p_{\text{max}}}{N\lambda_r}-\frac{1}{N}\sum\limits^{N}_{n=1}\frac{\gamma_r{\widetilde{\sigma}^2}}{{\beta}_n^2g_n}}, \, \forall n=1,\ldots, N,\\
		&\widetilde{p}_{\text{radar},n}= \sqrt{ \frac{\gamma_r{\widetilde{\sigma}^2}}{{\beta}_n^2g_n}}, \, \forall n=1,\ldots, N.
		\end{align}
	\end{subequations}
\end{theorem}
\begin{proof}
	The desired results can be obtained by following the same procedures shown in Appendix A and B.
\end{proof}

Based on (\ref{close2}), since the radar detection performance is not the optimization objective, the probe signal is still allocated the exact power to reach the performance requirement in the jamming power maximization problem (\ref{jam}). However, (\ref{close2}a) shows a different result compared to (\ref{close1}a). In the power minimization problem (\ref{power}), falls under case with enough energy, allocating more power for the probe signal indicates an increase in the radar power leakage, which helps jamming. On the contrary, in the jamming power maximization problem, allocating more power for the probe signal directly results in a decrease of the jamming power due to the limited total power. Moreover, (\ref{close1}a) indicates that the jamming power is limited by an upper bound determined by the minimum SINR requirement at the suspicious receiver, but (\ref{close2}a) shows that in the jamming power maximization problem, all power is allocated for jamming, except that which meets the radar detection performance requirement. The looser the radar detection performance constraint, the more power can be utilized to jam the suspicious receiver, which leads to a higher eavesdropping success probability, correspondingly.

	\vspace{-2ex}
	
\subsection{Overall Algorithm and Complexity Analysis}
In summary, since the original problem (\ref{main}) can be classified into two cases, there is no effect on the assumption that it is equivalent to the power minimization problem (\ref{power}) at first. To solve the power minimization problem (\ref{power}), we propose an orthogonal basis construction based algorithm and obtain the optimal solution in closed-form, which is shown in (\ref{close1}). To verify the assumption, the threshold $p_{\text{th}}$ calculated by (\ref{ptotal}) is compared with the power consumption upper bound $p_{\text{max}}$. If $p_{\text{max}} \geq p_{\text{th}}$, then $p_{\text{th}}$ is accessible and the obtained solution not only is optimal but also results in the minimum total power consumption. However, if $p_{\text{max}} < p_{\text{th}}$, then the assumption does not hold and the original problem is transformed into the jamming power maximization problem (\ref{jam}). By employing the orthogonal basis construction based algorithm, the solution in (\ref{close2}) is obtained, and it is also the optimal solution to the original problem. The overall algorithm for solving problem (\ref{main}) is summarized in Algorithm 1.
\begin{algorithm}[h]
	\caption{Orthogonal Basis Construction Based Algorithm for Problem (\ref{main})}
	\begin{algorithmic}[1]
		\State {\bf Initialize}:\ Calculate parameters $g_{\text{sum},n}$, $g_{\text{radar},n}$, and $g_{\text{jam}}$, $\forall n=1,\ldots, N$.
		\State Solve the power minimization problem (\ref{power}) and obtain the optimal solution according to (\ref{close1}).
		\State Calculate $p_{\text{th}}$ via (\ref{ptotal}) and compare it with $p_{\text{max}}$.
		\State {\bf if} $p_{\text{max}} \geq p_{\text{th}}$
		\State \quad Output the optimal solution of problem (\ref{main}).
		\State {\bf else}
		\State \quad Solve the jamming power maximization problem (\ref{jam}) and obtain the optimal solution based on (\ref{close2}).
		\State \quad Output the optimal solution of problem (\ref{main}).
		\State {\bf end}
		\label{algorithm 1}
	\end{algorithmic}
\end{algorithm}

We analyze the computational complexity of the proposed algorithm by evaluating the number of required
multiplications. Since all solutions are in closed-form, the computational complexity mainly comes from the matrix multiplications. Calculations of $g_{\text{sum},n}$, $g_{\text{radar},n}$, and $g_{\text{jam}}$, need $O(N^2M)$, $O(NM)$, and $O(N^2M)$ multiplications, respectively. The numbers of multiplications for computing (\ref{close1}a) and (\ref{close1}b) are $O(N^2)$ and $O(N)$, respectively. Then, since $p_{\text{total}}$ is calculated by accumulating $N$ square terms, it costs only $O(N)$ multiplications. Finally, the computation of the solutions based on (\ref{close2}) requires the same multiplications. Therefore, the total computational complexity of the proposed algorithm is $O(N^2M)$.

\section{Performance analysis}

The optimal solutions to the eavesdropping success probability maximization problem (\ref{p1}) are obtained in closed-form. Then, we focus on $\mathbb{E}\{Y\}=\mathcal{P}(\text{SINR}_\text{E}\ge \text{SINR}_\text{D})$ and analyze the eavesdropping success probability. We assume that $h_{sd}$ follows the zero-mean complex Gaussian random distribution with variance $\rho_{sd}$, the first $M-1$ elements of ${\bm \Sigma}\widetilde{{\bm U}}_n^H{\bm h}_{se}$ are identically and independently distributed (i.i.d.) zero-mean complex Gaussian random variables with variance $\rho_{se}$, and the first $M-1$ elements of ${\bm \Sigma}{\bm U}_n^H{\bm h}_{ed}$ are i.i.d. zero-mean complex Gaussian random variables with variance $\rho_{ed}$. The variances $\rho_{sd}$, $\rho_{se}$, and $\rho_{ed}$ represent the quality of each corresponding channel.

By following Lemma 1, $\text{SINR}_\text{E}$ can be simplified as
\begin{align}
\text{SINR}_\text{E}=\frac{p_s\|{\bm \Sigma}\widetilde{{\bm U}}_n^H{\bm h}_{se}\|^2}{\widetilde{\sigma}^2}.
\end{align}

Obviously, $\|{\bm \Sigma}\widetilde{{\bm U}}_n^H{\bm h}_{se}\|^2$ follows the Chi-square distribution with $2(M-1)$ degrees of freedom. Thus, the probability density function (PDF) of $\|{\bm \Sigma}\widetilde{{\bm U}}_n^H{\bm h}_{se}\|^2$ is given by
\begin{align}\label{fse}
f_{se}(x) = \frac{x^{M-2}}{\rho_{se}^{M-1}{\Gamma}(M-1)} e^{-\frac{x}{\rho_{se}}}.
\end{align}

In the power minimization problem (\ref{power}), $\text{SINR}_\text{D}$ equals $\gamma_s$, and the eavesdropping success probability can be transformed as
\begin{align}
\mathbb{E}\{Y\} = \mathcal{P}\left({\frac{p_s\|{\bm \Sigma}\widetilde{{\bm U}}_n^H{\bm h}_{se}\|^2}{\widetilde{\sigma}^2} \geq \gamma_s}\right).
\end{align}

Based on (\ref{fse}), the exact eavesdropping success probability is given by
\begin{align}\label{prob1}
\mathbb{E}\{Y\} = \int_{\frac{\widetilde{\sigma}^2\gamma_s}{p_s}}^{\infty}f_{se}(x)\text{d}x.
\end{align}

This shows that the eavesdropping success probability increases as the suspicious transmitter power increases, the noise decreases or $\text{SINR}_\text{D}$ decreases.

On the other hand, in the jamming power maximization problem (\ref{jam}), $\text{SINR}_\text{D}$ is given by
\begin{align}
\text{SINR}_\text{D}=\frac{|h_{sd}|^2p_s}{P_{\text{jam}}+\sigma^2},
\end{align}
where $P_{\text{jam}}= P_J g_{\text{jam}} + \dfrac{\gamma_r{\widetilde{\sigma}^2}}{{\beta}_n^2g_n}g_{\text{radar},n}$ and $P_J=\dfrac{p_{\text{max}}}{N\lambda_r}-\dfrac{1}{N}\sum\limits^{N}_{n=1}\dfrac{\gamma_r{\widetilde{\sigma}^2}}{{\beta}_n^2g_n}$.

Since $g_{\text{radar},n}$ is usually much smaller than $g_{\text{jam}}$, $g_{\text{radar},n}$ can be neglected, thus $P_{\text{jam}}= P_J g_{\text{jam}}$. Besides, since $g_{\text{jam}}={\|{\bm \Sigma}{\bm U}_n^H{\bm h}_{ed}\|}^2$, we assume that $g_{\text{jam}}$ follows the Chi-square distribution with $2(M-1)$ degrees of freedom. Similarly, the PDF of $g_{\text{jam}}$ can be written as
\begin{align}\label{fed}
f_{ed}(x) = \frac{x^{M-2}}{\rho_{ed}^{M-1}{\Gamma}(M-1)} e^{-\frac{x}{\rho_{ed}}}.
\end{align}

Meanwhile, $|h_{sd}|^2$ follows the exponential distribution with mean $\rho_{sd}$, and its PDF is given by
\begin{align}\label{fsd}
f_{sd}(x) = \frac{1}{\rho_{sd}} e^{-\frac{x}{\rho_{sd}}}.
\end{align}

Therefore, we have the following Lemma.
\begin{lemma} \label{lemma4}
	In the jamming power maximization problem (\ref{jam}), the exact eavesdropping success probability is given by
	\begin{align}
	\!\mathbb{E}\{Y\} =& 1-\sum\limits^{M-2}_{k=0}\frac{(-1)^k}{k!(M-k-2)}\left(\frac{\rho_{sd}\widetilde{\sigma}^2}{\rho_{se}\rho_{ed}P_J}\right)^{M-1}\notag\\
	&\times \hspace{-0.5ex}\text{exp}\left(\frac{{\sigma}^2}{\rho_{ed}P_J}\hspace{-0.5ex}+\hspace{-0.5ex}\frac{\rho_{sd}\widetilde{\sigma}^2}{\rho_{se}\rho_{ed}P_J}\right) \hspace{-0.5ex} \left(\frac{{\sigma}^2}{\rho_{ed}P_J}\hspace{-0.5ex}+\hspace{-0.5ex}\frac{\rho_{sd}\widetilde{\sigma}^2}{\rho_{se}\rho_{ed}P_J}\right)^k \notag\\
	&\times \hspace{-0.5ex}\Gamma \left(-k,\frac{{\sigma}^2}{\rho_{ed}P_J}\hspace{-0.5ex}+\hspace{-0.5ex}\frac{\rho_{sd}\widetilde{\sigma}^2}{\rho_{se}\rho_{ed}P_J}\right).\!
\end{align}
\end{lemma}
\begin{proof}
	See Appendix C.
\end{proof}	

The above result indicates that the eavesdropping success probability is positively correlated with the jamming power $P_J$ and negatively correlated with the channel quality parameter $\rho_{sd}$.

\section{simulation results}
In this section, simulation results are presented to demonstrate the effectiveness of the proposed algorithm. Unless otherwise specified, the numbers of transmit/receive antennas and RF chains at the legitimate monitor are $N=128$ and $M=4$, respectively. Besides, the noise variances are normalized such that $\sigma^2=1$ and $\widetilde{\sigma}^2=2$. The parameters related to the SINR constraints are set as $\gamma_s=0\ \text{dB}$ and $\gamma_r=10\ \text{dB}$, and the SNR of the suspicious transmitter and the overall system are given by $\dfrac{p_s} {\sigma^2}=10\ \text{dB}$ and $\dfrac{p_{\text{max}}} {N\sigma^2}=20\ \text{dB}$, respectively. In addition, we set the channel quality parameters $\rho_{sd}=10$, $\rho_{ed}=1$, and $\rho_{se}=1$ and the time ratios $\lambda_r=0.1$ and $\lambda_w=0.9$, respectively.

We compare the performance based on the following two benchmark algorithms.
\begin{itemize}
	\item[$\bullet$]Maximum-ratio combining (MRC) algorithm \cite{pe4}. The legitimate surveillance system and the radar function both design their digital receive beamforming vectors by using the MRC algorithm. The MRC algorithm directly maximizes the receiving signal gains, but the interference between the legitimate surveillance system and the radar function is not eliminated. Specifically, the MRC algorithm uses ${\bm w}_{s,n} = \dfrac{\widetilde{{\bm U}}_n^H{\bm h}_{se}}{\|\widetilde{{\bm U}}_n^H{\bm h}_{se}\|}$ and ${\bm w}_{r,n} = \dfrac{\widetilde{{\bm U}}_n^H{\bm A}_n{\bm U}_n{\bm p}_{n,r}}{\|\widetilde{{\bm U}}_n^H{\bm A}_n{\bm U}_n{\bm p}_{n,r}\|}$.
	
	\item[$\bullet$]Surveillance centric algorithm. In this algorithm, the CSI for the radar function is unknown to the legitimate surveillance system. Thus, the digital receive beamforming vector for proactive eavesdropping is designed by employing the MRC algorithm, while the digital receive beamforming vector for radar detection is still designed in the null-space of the surveillance signal.
\end{itemize}

\begin{figure}[t]
	\vspace{-2ex}
	\centering
	\includegraphics[width=0.8\linewidth]{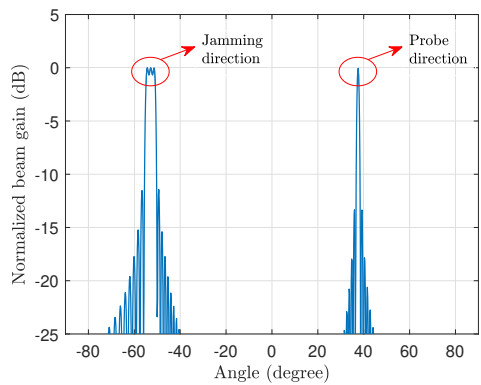}
	\vspace{-1ex}
	\caption{Beampattern of a specific transmit signal before power allocation.}
	\label{beam}
	\vspace{-2ex}
\end{figure}

We first show the transmit signal beampattern of the legitimate monitor in Fig. \ref{beam}. It can be observed that the beam gain mainly focuses on two directions, which indicates simultaneous jamming and radar detection. Besides, it is obvious that there are $3$ beam gain peaks aiming at the jamming direction and only $1$ peak aiming at the probe direction. The reason is that, $3$ RF chains are utilized for surveillance and $1$ RF chain is used for radar detection when there are $M=4$ RF chains. In addition, it is noted that, due to the large number of antennas, the beam gain aiming at non-target directions is quite low, leading to the high spatial resolution of the transmit signal.

\begin{figure}[t]
	\vspace{-2ex}
	\centering
	\includegraphics[width=0.8\linewidth]{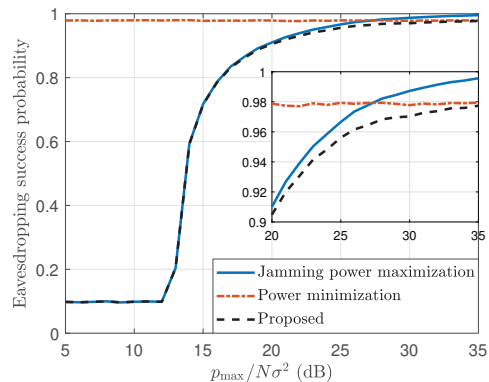}
	\vspace{-1ex}
	\caption{The switching between the solutions of the two cases.}
	\label{switch}
	\vspace{-2ex}
\end{figure}

In Fig. \ref{switch}, the classification of the two subproblems and the switching between the solutions corresponding to the two cases are depicted. Note that the solutions obtained by the jamming power maximization scheme are only corresponding to the jamming power maximization problem and the solutions obtained by the power minimization scheme are only corresponding to the power minimization problem, while the proposed scheme is a combination of these two benchmark schemes and obtains the solutions based on correct classification.	It can be seen that for the jamming power maximization scheme, the eavesdropping success probability continuously increases with the increase of $p_{\text{max}}$, due to the increase in jamming power. On the other hand, for the power minimization scheme, the eavesdropping success probability remains constant since it assumed that $\text{SINR}_\text{D}$ always equals $\gamma_s$. When $p_{\text{max}}$ is low, i.e., the transmit power is inadequate, the eavesdropping success probability maximization problem should be classified as the jamming power maximization problem, and the corresponding solutions should be applied. As $p_{\text{max}}$ keeps increasing and exceeds the threshold, the proposed algorithm is switched into the case corresponding to the power minimization problem and applies its solutions. Meanwhile, it can be observed that the two benchmark schemes achieve better performance since they may not find the feasible solution to the eavesdropping success probability maximization problem. The solutions obtained by the proposed scheme are always based on correct classification results, but the solutions obtained by the jamming power maximization scheme are only corresponding to the jamming power maximization problem even for the power minimization problem, as well as the power minimization scheme. Therefore, the solutions obtained by the two benchmark schemes may be based on the incorrect classification, and they are not feasible since they do not satisfy all the constraints, thus leading to an extra performance gain.

\begin{figure}[t]
	\vspace{-2ex}
	\centering
	\includegraphics[width=0.8\linewidth]{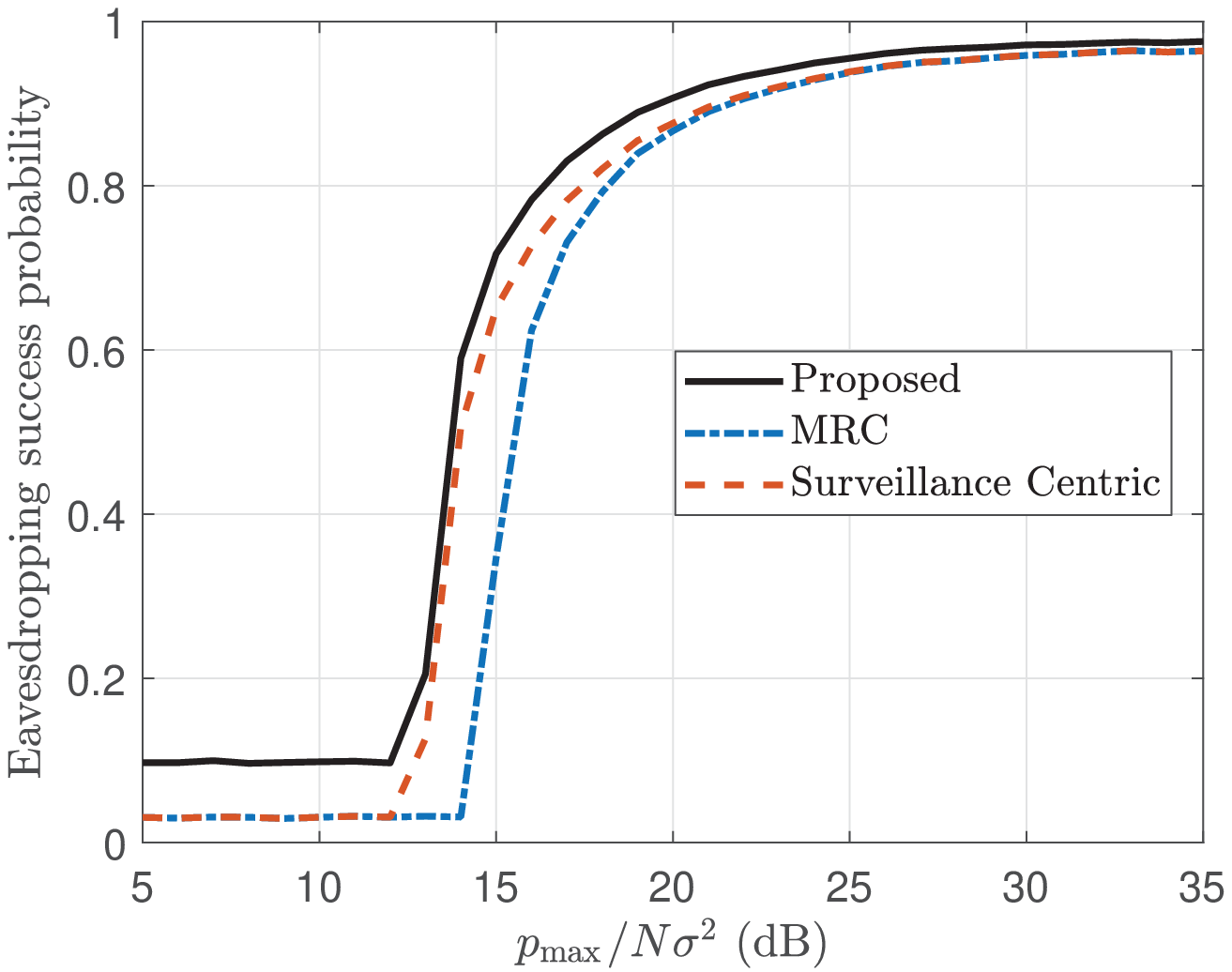}
	\vspace{-1ex}
	\caption{Eavesdropping success probability comparison of the proposed, MRC, and surveillance centric algorithms versus ${p_{\text{max}}}$.}
	\label{pmax}
	\vspace{-2ex}
\end{figure}

Fig. \ref{pmax} compares the eavesdropping success probability of the proposed, MRC, and surveillance centric algorithms versus ${p_{\text{max}}}$. It can be seen that our proposed algorithm achieves higher eavesdropping success probabilities than other baseline algorithms. Besides, the simulation results also indicate the impact of ${p_{\text{max}}}$ on the eavesdropping success probability. First, when ${p_{\text{max}}}$ is quite low, the eavesdropping success probability of these three algorithms is kept at a certain low level. This is because when the transmit power is insufficient, all power is allocated to the probe signal to guarantee the performance of radar detection. Thus, since the suspicious transmission usually occupies a better channel, it is difficult to successfully eavesdrop the suspicious transmission without jamming. Meanwhile, because of the null-space interference elimination approach, our proposed algorithm still outperforms the compared algorithms. With the increase of ${p_{\text{max}}}$, more and more power is allocated to the probe signal until the detection performance meets the requirement. Then, the rest of the transmit power starts to be allocated for jamming the suspicious receiver and, as a consequence, the eavesdropping success probability increases. Finally, when the transmit power becomes sufficient, limited by the minimum essential monitoring rate constraint, $\text{SINR}_\text{D}$ reaches $\gamma_s$ and the eavesdropping success probability converges to a certain high level.

\begin{figure}[t]
	\vspace{-2ex}
	\centering
	\includegraphics[width=0.8\linewidth]{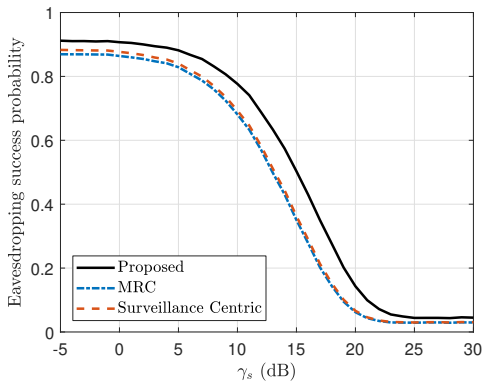}
		\vspace{-1ex}
	\caption{Eavesdropping success probability comparison of the proposed, MRC, and surveillance centric algorithms versus $\gamma_s$.}
	\label{gamma}
	\vspace{-2ex}
\end{figure}

Fig. \ref{gamma} depicts the relation between the eavesdropping success probability and the minimum essential monitoring rate $\gamma_s$. We observe that when $\gamma_s$ increases, the eavesdropping success probability first remains nearly unchanged and then decreases rapidly. The reason is that when $\gamma_s$ is quite low, it is unreachable for $\text{SINR}_\text{D}$ and does not affect the system. Hence, the system tends to utilize all power to decrease the rate at the suspicious receiver, which leads to the eavesdropping success probability being close to $100\%$. However, when $\gamma_s$ keeps increasing, it becomes accessible and the jamming power becomes limited. Therefore, with the increase of $\gamma_s$, the jamming power decreases, which results in the reduction of the eavesdropping success probability. Besides, it can be observed that our proposed algorithm provides the best performance among the analyzed algorithms when $\gamma_s$ increases.

\begin{figure}[t]
	\vspace{-2ex}
	\centering
	\includegraphics[width=0.8\linewidth]{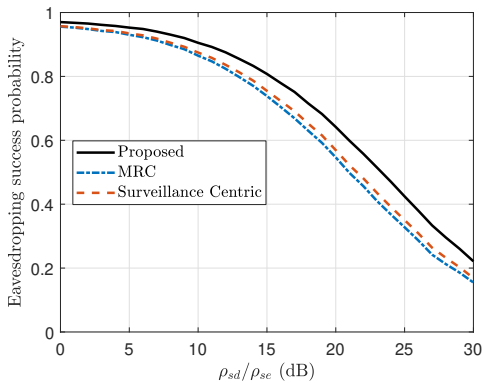}
		\vspace{-1ex}
	\caption{Eavesdropping success probability comparison of the proposed, MRC, and surveillance centric algorithms versus $\rho_{sd}$.}
	\label{rho}
	\vspace{-2ex}
\end{figure}

Fig. \ref{rho} depicts the impact of the channel quality on the eavesdropping success probability, where $\dfrac{\rho_{sd}}{\rho_{se}}$ represents the quality of the suspicious channel for legitimate eavesdropping. Generally, the legitimate monitor may be far away from the suspicious transmitter to avoid getting exposed, which leads to a worse quality of the legitimate eavesdropping channel compared to the suspicious transmission channel. Thus, with the increase of $\dfrac{\rho_{sd}}{\rho_{se}}$, the suspicious channel quality becomes better and eavesdropping becomes more difficult. The simulation results also verify the negative correlation between the eavesdropping success probability and $\dfrac{\rho_{sd}}{\rho_{se}}$. Besides, we can again see the superior performance of our proposed algorithm.

\begin{figure}[t]
	\vspace{-2ex}
	\centering
	\includegraphics[width=0.8\linewidth]{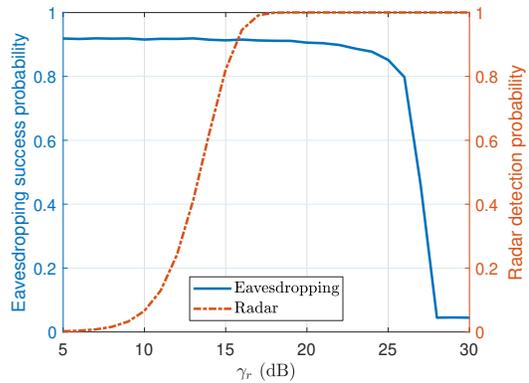}
		\vspace{-1ex}
	\caption{The trade-off between the eavesdropping success probability and radar detection probability versus $\gamma_r$.}
	\label{radar}
	\vspace{-2ex}
\end{figure}

Finally, in Fig. \ref{radar}, we illustrate the trade-off between the legitimate surveillance system and the radar function. The detection probability of the radar function is obtained and shown in Fig. \ref{radar}, according to \cite{radarprob1} and \cite{radarprob2}. We can observe that with the increase of $\gamma_r$, the radar detection probability increases from 0 to 1 and then stays constant. However, in contrast, the eavesdropping success probability will decrease rapidly if $\gamma_r$ reaches a certain high level. In fact, Fig. \ref{radar} shows the coexistence between the legitimate surveillance and the radar function. For the radar function, a larger $\gamma_r$ corresponds to a higher $\text{SINR}_\text{R}$ and leads to a greater radar detection probability. On the other hand, since the total power is finite, a larger $\gamma_r$ also indicates less jamming power and lower eavesdropping success probability. Therefore, $\gamma_r$ should be determined by balancing the eavesdropping success probability and the radar detection probability.

\section{conclusion}
In this paper, we have investigated a wireless legitimate surveillance system with radar function, in which we aimed at maximizing the eavesdropping success probability. We formulated the optimization problem by transforming the eavesdropping success probability into the difference of SINRs subject to the performance requirements of radar detection and surveillance. To tackle this challenging problem, we employed the Rayleigh quotient to simplify the problem and obtained the optimal receive beamforming vectors in closed-form. Then, we applied the divide-and-conquer principle to decompose the problem into two subproblems according to two different cases. The first subproblem aims at minimizing the total transmit power and the second subproblem focuses on maximizing the jamming power. For both subproblems, the optimal digital transmit beamforming vector was obtained in closed-form with the help of orthogonal decomposition. Combining the two cases, we developed the optimal beamforming algorithm. Additionally, some insightful results regarding the optimal transmit power allocation, the threshold between the two cases, and the probability analysis have also been discussed. Our simulation results demonstrated that the proposed algorithm achieves better performances than key baseline algorithms.

\vspace{-4.5mm}

\begin{appendices}
	\section{Proof of lemma 3}
Based on (\ref{bases}) and (\ref{basesip}), we construct another orthogonal basis ${\bm v}_{\text{new},n}$ as follows
\begin{align}\label{vnew}
{\bm v}_{\text{new},n} = -\sqrt{\frac{g_{\text{radar},n}}{g_{\text{jam}}}}{\bm v}_{\text{sum},n}+\sqrt{\frac{g_{\text{sum},n}}{g_{\text{jam}}}}{\bm v}_{\text{radar}},
\end{align}
where ${\bm v}_{\text{new},n}$ is normalized and orthogonal to ${\bm v}_{\text{sum},n}$.

Then, ${\bm p}_{n,r}$ and $\widetilde{\bm p}_{n,w}$ can be decomposed by ${\bm v}_{\text{sum},n}$, ${\bm v}_{\text{new},n}$, and ${\bm v}_{\text{jam}}$ as
\begin{align}\label{pdcmp1}
	{\bm p}_{n,r}&=p_{\text{sum},n}{\bm v}_{\text{sum},n}+p_{\text{new},n}{\bm v}_{\text{new},n}+{\bm r}_{1,n},\notag \\ \widetilde{\bm p}_{n,w}&=p_{\text{jam},n}{\bm v}_{\text{jam}}+{\bm r}_{0,n},
\end{align}
where ${\bm r}_{1,n}$ and ${\bm r}_{0,n}$ represent the components orthogonal to the subspace determined by ${\bm v}_{\text{sum},n}$ and ${\bm v}_{\text{new},n}$, respectively. Besides, $p_{\text{sum},n},p_{\text{new},n}$, and $p_{\text{jam},n}$ are scalars denoting
the modulus of each direction basis.

By employing (\ref{pdcmp1}), problem (\ref{power}) can be simplified into the following problem:
\begin{subequations}\label{p4}
	\begin{eqnarray}
		&\hspace{-4ex}\min\limits_{\left\{\substack{p_{\text{sum},n},p_{\text{new},n},\\p_{\text{jam},n},{\bm r}_{1,n},\\{\bm r}_{0,n}}\right\} }&\hspace{-2ex}\sum\limits^{N}_{n=1}(\lambda_r(p_{\text{sum},n}^2+p_{\text{new},n}^2 +\|{\bm r}_{1,n}\|^2)\notag \\
		&&\hspace{-4ex}+\lambda_w(p_{\text{jam},n}^2+\|{\bm \Sigma}{\bm r}_{0,n}\|^2))\\
		&\hspace{-4ex}\text{s.t.}&\hspace{-2ex}\sum\limits^{N}_{n=1}\hspace{-0.5ex}(\lambda_rp_{\text{sum},n}^2g_{\text{sum},n}\hspace{-0.5ex}+\hspace{-0.5ex}\lambda_wp_{\text{jam},n}^2g_{\text{jam}})\hspace{-0.5ex}=\hspace{-0.5ex} C_1,\notag \\
		&&\hspace{-2ex}\\
		&&\hspace{-2ex}\sqrt{\frac{g_{\text{radar},n}}{g_{\text{sum},n}}}p_{\text{sum},n}\hspace{-0.5ex}+\hspace{-0.5ex}\sqrt{\hspace{-0.5ex}\frac{g_{\text{jam}}}{g_{\text{sum},n}}}p_{\text{new},n}\hspace{-0.5ex}\geq \hspace{-0.5ex}C_{2,n}, \notag \\
		&&\hspace{-2ex} ~~~~ \forall n\hspace{-0.5ex}=\hspace{-0.5ex}1,...,N,\\
		&&\hspace{-2ex}p_{\text{sum},n}\geq 0,p_{\text{new},n}\geq 0,p_{\text{jam},n}\geq 0,\notag \\
		&&\hspace{-2ex} ~~~~ \forall n\hspace{-0.5ex}=\hspace{-0.5ex}1,...,N,
	\end{eqnarray}
\end{subequations}
where $C_1 \triangleq \dfrac{N|h_{sd}|^2p_s}{\gamma_s}-N{\sigma}^2$ and $C_{2,n} \triangleq \sqrt{\dfrac{\gamma_r{\widetilde{\sigma}^2}}{{\beta}_n^2g_n}}, \forall n=1,\ldots, N$.

Apparently, ${\bm r}_{1,n}$ and ${\bm r}_{0,n}$ only exist in (\ref{p4}a), where their powers are required to be minimized, thus it is always optimal to set ${\bm r}_{1,n}=\bm{0}$ and ${\bm r}_{0,n}=\bm{0}$.
Besides, since constraint (\ref{p4}b) is non-convex, to proceed, we employ the following variable substitutions:
$x_{\text{sum},n}=p_{\text{sum},n}^2$, $x_{\text{new},n}=p_{\text{new},n}^2$, and $x_{\text{jam},n}=p_{\text{jam},n}^2.$ Then, problem (\ref{p4}) is equivalently transformed into a convex problem as follows
\begin{subequations}\label{px1}
	\begin{eqnarray}
		&\hspace{-2ex}\min\limits_{\left\{\substack{x_{\text{sum},n},x_{\text{new},n},\\x_{\text{jam},n}}\right\} }&\hspace{-2ex}\sum\limits^{N}_{n=1}(\lambda_r(x_{\text{sum},n}+x_{\text{new},n})+\lambda_wx_{\text{jam},n})\\
		&\hspace{-2ex}\text{s.t.}&\hspace{-2ex}\sum\limits^{N}_{n=1}(\lambda_rx_{\text{sum},n}g_{\text{sum},n}+\lambda_wx_{\text{jam},n}g_{\text{jam}})= C_1,\notag \\
		&& \hspace{-2ex}\\
		&&\hspace{-2ex}\sqrt{\frac{g_{\text{radar},n}}{g_{\text{sum},n}}}\sqrt{x_{\text{sum},n}}+\sqrt{\frac{g_{\text{jam}}}{g_{\text{sum},n}}}\sqrt{x_{\text{new},n}}\geq C_{2,n},\notag \\
		&&\hspace{-2ex} ~~~~\forall n=1,...,N,\\
		&&\hspace{-2ex}x_{\text{sum},n}\geq 0,x_{\text{new},n}\geq 0,x_{\text{jam},n}\geq 0,\notag \\
		&&\hspace{-2ex}~~~~\forall n=1,...,N.
	\end{eqnarray}
\end{subequations}

The Lagrangian multiplier function for problem (\ref{px1}) can be expressed as
\begin{align}
	\!\mathcal L\
	=&\sum\limits^{N}_{n=1}(\lambda_r(x_{\text{sum},n}+x_{\text{new},n})+\lambda_wx_{\text{jam},n})\notag\\
	&-\sum\limits^{N}_{n=1}(\mu_{1,n}x_{\text{sum},n}+\mu_{2,n}x_{\text{new},n}+\mu_{3,n}x_{\text{jam},n})\notag\\
	&+\mu_4 (\sum\limits^{N}_{n=1}(\lambda_rx_{\text{sum},n}g_{\text{sum},n}+\lambda_wx_{\text{jam},n}g_{\text{jam}})-C_1)\notag\\
	&-\sum\limits^{N}_{n=1}\hspace{-0.5ex}\mu_{5,n}\hspace{-0.5ex}\left(\hspace{-0.5ex}\sqrt{\frac{g_{\text{radar},n}}{g_{\text{sum},n}}}\sqrt{x_{\text{sum},n}}\right. \left.\hspace{-0.5ex}+\hspace{-0.5ex}\sqrt{\frac{g_{\text{jam}}}{g_{\text{sum},n}}}\sqrt{x_{\text{new},n}}\hspace{-0.5ex}-\hspace{-0.5ex}C_{2,n}\hspace{-0.5ex}\right),\!
\end{align}
where $\mu_{1,n}\geq 0,\mu_{2,n}\geq 0,\mu_{3,n}\geq 0,\mu_{4} \ne 0$, and $\mu_{5,n}\geq 0$ denote the dual variables of problem (\ref{px1}) associated with the constraints in (\ref{px1}b), (\ref{px1}c), and (\ref{px1}d), respectively. As problem (\ref{px1}) is convex and satisfies Slater's condition, the KKT conditions are necessary and sufficient for establishing the optimal solution. The KKT conditions for problem (\ref{px1}) are given by
\begin{subequations}\label{kkt1}
	\begin{eqnarray}
		&&\hspace{-2ex}\frac{\partial \mathcal L}{\partial x_{\text{sum},n}}\hspace{-0.5ex}=\hspace{-0.5ex}\lambda_r \hspace{-0.5ex}-\hspace{-0.5ex}\mu_{1,n}\hspace{-0.5ex}+\hspace{-0.5ex}\mu_4 \lambda_rg_{\text{sum},n} \hspace{-0.5ex}-\hspace{-0.5ex}\frac{\mu_{5,n}\sqrt{g_{\text{radar},n}}}{2\sqrt{g_{\text{sum},n}x_{\text{sum},n}}}\hspace{-0.5ex}=\hspace{-0.5ex}0,\notag\\
		&&\hspace{7.5ex}\forall n=1,...,N,\\
		&&\hspace{-2ex}\frac{\partial \mathcal L}{\partial x_{\text{jam},n}}\hspace{-0.5ex}=\hspace{-0.5ex}\lambda_w\hspace{-0.5ex}-\hspace{-0.5ex}\mu_{3,n}\hspace{-0.5ex}+\hspace{-0.5ex}\mu_4 \lambda_wg_{\text{jam}} \hspace{-0.5ex}=\hspace{-0.5ex}0,\forall n=1,...,N,\\
		&&\hspace{-2ex}\mu_{3,n}x_{\text{jam},n}=0,\forall n=1,...,N.
	\end{eqnarray}
\end{subequations}

As shown in (\ref{kkt1}c), either $\mu_{3,n}$ or $x_{\text{jam},n}$ must be equal to $0$. To reduce ambiguity, we assume that there exists a specific $n$ satisfying $\mu_{3,n}=0$. Then, in (\ref{kkt1}b), $\mu_{4}=-\dfrac{1}{g_{\text{jam}}}$ is obtained. Thus, in (\ref{kkt1}a), we have:
\begin{align} \label{contradiction1}
\frac{\lambda_r(g_{\text{jam}}-g_{\text{sum},n})}{g_{\text{jam}}}-\mu_{1,n}-\frac{\mu_{5,n}\sqrt{g_{\text{radar},n}}}{2\sqrt{g_{\text{sum},n}x_{\text{sum},n}}}=0.
\end{align}

However, since $\forall n=1,\ldots, N$, $g_{\text{jam}} < g_{\text{sum},n}$, $\lambda_r > 0$, $\mu_{1,n}\geq 0$, $\mu_{5,n}\geq 0$, and $\dfrac{\sqrt{g_{\text{radar},n}}}{2\sqrt{g_{\text{sum},n}x_{\text{sum},n}}} > 0$, (\ref{contradiction1}) cannot be satisfied, thus the assumption $\mu_{3,n}=0$ is incorrect for all $n$, which means that $x_{\text{jam},n}=0,\forall n=1,\ldots, N.$

Therefore, we prove that in problem (\ref{power}), ${\bm p}_{n,r}$ has only components lying in the subspace determined by ${\bm v}_{\text{radar}}$ and ${\bm v}_{\text{jam}}$, and $\widetilde{\bm p}_{n,w}=\bm{0},\forall n=1,\ldots, N.$ This completes the proof.

\vspace{-1ex}
\section{Proof of theorem 1}

For problem (\ref{power2}), we employ the following variable substitutions: $x_{\text{jam},n}=p_{\text{jam},n}^2$ and $x_{\text{radar},n}=p_{\text{radar},n}^2$.

Similar to problem (\ref{px1}), the problem is transformed into a convex problem as
\begin{subequations}\label{px2}
	\begin{eqnarray}
		&\hspace{-6ex}\min\limits_{\{x_{\text{jam},n},x_{\text{radar},n}\} }&\hspace{-2ex}\sum\limits^{N}_{n=1}\lambda_r(x_{\text{jam},n}+x_{\text{radar},n})\\
		&\hspace{-6ex}\text{s.t.}&\hspace{-2ex}\sum\limits^{N}_{n=1}\lambda_r(x_{\text{jam},n}g_{\text{jam}}\hspace{-0.5ex}+\hspace{-0.5ex}x_{\text{radar},n}g_{\text{radar},n})= C_1,\\
		&&\hspace{-2ex}x_{\text{radar},n}\geq C_{2,n}^2, \forall n=1,...,N,\\
		&&\hspace{-2ex}x_{\text{jam},n}\geq 0,x_{\text{radar},n}\geq 0,\forall n=1,...,N.
	\end{eqnarray}
\end{subequations}

\vspace{-1ex}
The Lagrangian multiplier function for problem (\ref{px2}) can be expressed as
\begin{align}
	\mathcal L &=\sum\limits^{N}_{n=1}\lambda_r(x_{\text{jam},n}+x_{\text{radar},n}) -\sum\limits^{N}_{n=1}(\mu_{1,n}x_{\text{jam},n}+\mu_{2,n}x_{\text{radar},n})\notag\\
	&+\mu_3 (\sum\limits^{N}_{n=1}\lambda_r(x_{\text{jam},n}g_{\text{jam}}+x_{\text{radar},n}g_{\text{radar},n})-C_1)\notag \\
	&-\sum\limits^{N}_{n=1}\mu_{4,n}(x_{\text{radar},n}-C_{2,n}^2),
\end{align}
where $\mu_{1,n}\geq 0,\mu_{2,n}\geq 0,\mu_{3} \ne 0$, and $\mu_{4,n}\geq 0$ denote the dual variables associated with (\ref{px2}b), (\ref{px2}c), and (\ref{px2}d), respectively. The KKT conditions for problem (\ref{px2}) are given by
\vspace{-1ex}
\begin{subequations}\label{kkt2}
	\begin{eqnarray}
		&&\hspace{-8ex}\frac{\partial \mathcal L}{\partial x_{\text{jam},n}}\hspace{-0.8ex}=\hspace{-0.5ex}\lambda_r\hspace{-0.5ex}-\hspace{-0.5ex}\mu_{1,n} \hspace{-0.5ex}+\hspace{-0.5ex}\mu_3\lambda_rg_{\text{jam}}\hspace{-0.5ex}=\hspace{-0.8ex}0,\forall n\hspace{-0.8ex}=\hspace{-0.5ex}1,...,N,\\
		&&\hspace{-8ex}\frac{\partial \mathcal L}{\partial x_{\text{radar},n}}\hspace{-0.8ex}=\hspace{-0.5ex}\lambda_r\hspace{-0.5ex}-\hspace{-0.5ex}\mu_{2,n} \hspace{-0.5ex}+\hspace{-0.5ex}\mu_3\lambda_rg_{\text{radar},n}\hspace{-0.5ex}-\hspace{-0.5ex}\mu_{4,n}\hspace{-0.8ex}=\hspace{-0.5ex}0,\notag \\
		&&\hspace{2ex}\forall n\hspace{-0.8ex}=\hspace{-0.5ex}1,...,N,\\
		&&\hspace{-8ex}\mu_{2,n}x_{\text{radar},n}=0,\forall n=1,...,N,\\
		&&\hspace{-8ex}\mu_{4,n}(x_{\text{radar},n}-C_{2,n}^2)=0,\forall n=1,...,N.
	\end{eqnarray}
\end{subequations}

\vspace{-2ex}
Based on (\ref{kkt2}), as $C_{2,n}^2=\dfrac{\gamma_r{\widetilde{\sigma}^2}}{{\beta}_n^2g_n} > 0$, (\ref{px2}c) indicates that $x_{\text{radar},n}\geq C_{2,n}^2 > 0, \forall n=1,...,N$, thus $\mu_{2,n}=0,\forall n=1,...,N$ can be obtained by employing (\ref{kkt2}c). Besides, in (\ref{kkt2}d), either $\mu_{4,n}$ or $x_{\text{radar},n}-C_{2,n}^2$ must be equal to 0. To reduce ambiguity, we assume that there exists a specific $n$ satisfying $\mu_{4,n}=0$. Then, in (\ref{kkt2}b), $\mu_{3}=-\dfrac{1}{g_{\text{radar},n}}$ is obtained. Thus, in (\ref{kkt2}a), we have:
\begin{align} \label{contradiction2}
\frac{\lambda_r(g_{\text{radar},n}-g_{\text{jam}})}{g_{\text{radar},n}}-\mu_{1,n}=0.
\end{align}

However, since $\forall n=1,\ldots, N$, $g_{\text{radar},n} < g_{\text{jam}}$, $\lambda_r > 0$, and $\mu_{1,n}\geq 0$, (\ref{contradiction2}) cannot be satisfied, thus the assumption $\mu_{4,n}=0$ does not hold for all $n$, which means that $x_{\text{radar},n}=C_{2,n}^2,\forall n=1,\ldots, N$. Based on (\ref{power2}b), the solutions of $p_{\text{jam},n}$ and $p_{\text{radar},n}$ are given by
\begin{subequations}
	\begin{eqnarray}
		&p_{\text{jam},n}=& \sqrt{\frac{|h_{sd}|^2p_s-{\sigma}^2\gamma_s}{\gamma_s\lambda_rg_{\text{jam}}}-\frac{1}{Ng_{\text{jam}}}\sum\limits^{N}_{n=1}\frac{\gamma_r{\widetilde{\sigma}^2}g_{\text{radar},n}}{{\beta}_n^2g_n}},\notag \\
		&&\forall n=1,...,N,\\
		&p_{\text{radar},n}=& \sqrt{ \frac{\gamma_r{\widetilde{\sigma}^2}}{{\beta}_n^2g_n}}, \forall n=1,...,N.
	\end{eqnarray}
\end{subequations}

This completes the proof.

\section{Proof of lemma 4}
The proof can be established based on a similar method in \cite{pe4}. Define $\gamma_{se}=\|{\bm \Sigma}\widetilde{{\bm U}}_n^H{\bm h}_{se}\|^2$, $\gamma_{sd}=|h_{sd}|^2$, and $\gamma_{ed}=g_{\text{jam}}.$ Then, the eavesdropping success probability can be written as
\begin{align}\label{apdx_c_1}
	\mathbb{E}\{Y\} = \mathcal{P}\left({\frac{p_s}{\widetilde{\sigma}^2} \gamma_{se} \geq \frac{\dfrac{p_s}{{\sigma}^2} \gamma_{sd}}{\dfrac{P_J}{{\sigma}^2} \gamma_{ed}+1}}\right).
\end{align}

Based on (\ref{fse}) and (\ref{fed}), conditioned on $\gamma_{se}$ and $\gamma_{ed}$ with the aid of \cite[eq. (3.351.1)]{prob24}, (\ref{apdx_c_1}) can be simplified to
\begin{align}\label{apdx_c_2}
	\mathbb{E}\{Y\} =1- \text{exp}\left({-\frac{\gamma_{se} {\sigma}^2}{\rho_{sd}\widetilde{\sigma}^2} \left(
		{
		\frac{P_J}{{\sigma}^2} \gamma_{ed} +1
	}
		\right)
		}\right).
\end{align}

Then, utilizing \cite[eq. (3.351.3)]{prob24}, (\ref{apdx_c_2}) is further simplified to
\begin{align}\label{apdx_c_3}
	\mathbb{E}\{Y\} =1- \left({1+\frac{\rho_{se} {\sigma}^2}{\rho_{sd}\widetilde{\sigma}^2} \left(
		{
			\frac{P_J}{{\sigma}^2} \gamma_{ed} +1
		}
		\right)
	}\right)^{-(M-1)}.
\end{align}

Applying (\ref{fed}), (\ref{apdx_c_3}) is transformed as
\begin{align}\label{apdx_c_4}
	\mathbb{E}\{Y\} =&1- \int_{0}^{\infty}	
	\frac{(\dfrac{\rho_{sd}\widetilde{\sigma}^2}{\rho_{se} P_J})^{M-1}}
	{(x+\dfrac{{\sigma}^2}{P_J}+\dfrac{\rho_{sd}\widetilde{\sigma}^2}{\rho_{se} P_J})^{M-1}} \notag\\
	&\times
	\frac{x^{M-2}}{\rho_{ed}^{M-1}{\Gamma}(M-1)} e^{-\frac{x}{\rho_{ed}}}
	\text{d}x.
\end{align}

Finally, employing $t=x+\dfrac{{\sigma}^2}{P_J}+\dfrac{\rho_{sd}\widetilde{\sigma}^2}{\rho_{se} P_J}$ and the binomial expansion, the desired result can be obtained by utilizing \cite[eq. (3.381.3)]{prob24}. This completes the proof.

\end{appendices}


\vspace{-2mm}
\bibliography{references}

\end{document}